\documentclass[runningheads]{llncs}
\usepackage{graphicx}
\usepackage[appendix=strip,bibliography=common]{apxproof}

\usepackage{graphicx}
\usepackage{amsmath,amssymb}
\usepackage{mathtools}
\mathtoolsset{showonlyrefs=false} 
\usepackage{cite}
\usepackage{caption}    
\usepackage[skip=0.5ex]{subcaption} 
\usepackage{url}
\usepackage{bm}
\usepackage{textcomp}

\def\id#1{\ensuremath{\mathit{#1}}}
\let\idit=\id

\def\idrm#1{\ensuremath{\mathrm{#1}}}
\def\idtt#1{\ensuremath{\mathtt{#1}}}

\usepackage{xsavebox}

\newsavebox{\cmbox}
\newenvironment{commbox}{
  \begin{lrbox}{\cmbox}
    \begin{minipage}{.9\textwidth}
}{\end{minipage}
  \end{lrbox}
}

\newenvironment{myempty}{\begin{commbox}}{\end{commbox}}

\newenvironment{inalign}{\begin{math}\catcode`&=9}{\catcode`&=4\end{math}}

\usepackage{cleveref}
\crefname{section}{Sec.}{Sections}
\crefname{chapter}{Chapter}{Chapters}
\crefname{algorithm}{Algorithm}{Algorithms}
\crefname{table}{Table}{Tables}
\crefname{figure}{Fig.}{Figures}
\crefname{definition}{Def.}{Definitions}
\crefname{lemma}{Lem.}{Lemmas}
\crefname{proposition}{Prop.}{Propositions}
\crefname{theorem}{Thm.}{Theorems}
\crefname{remark}{Remark}{Remarks}
\crefname{problem}{Problem}{Problems}
\crefname{observation}{Observation}{Observations}
\crefname{equation}{Eq.}{Equations}

\newtheoremrep{definition}{Definition}[section]
\newtheoremrep{lemma}[theorem]{Lemma}[section]
\newtheoremrep{proposition}[proposition]{Proposition}[section]
\newtheoremrep{theorem}[theorem]{Theorem}[section]
\newtheoremrep{corollary}[theorem]{Corollary}[section]
\newtheoremrep{remark}[theorem]{Remark}

\usepackage[inline,shortlabels]{enumitem} 
\setlist{noitemsep}
\setlist{%
topsep = 0.25\baselineskip,
parsep = 0.125\baselineskip
}
\setlist[enumerate]{ labelsep=.25pc, leftmargin=1.5pc } 
\setlist[enumerate,1]{ label= (\arabic*), ref=\arabic*}
\setlist[enumerate,2]{ label= (\roman*),ref  = \roman*}
\setlist[enumerate,3]{ label= (\alph*), ref  = (\alph*)}
\setlist[itemize]{ leftmargin=1.5pc }
\setlist[description]{ font=\sffamily\bfseries }

\renewcommand{\subsubsection}[1]{\textbf{#1}\hspace{0.25em}}

\usepackage{yhmath}

\renewcommand{\-}{\mbox{\textit{-}}}

\newcommand{\nat}{\mathbb{N}}

\renewcommand{\leq}{\leqslant}
\renewcommand{\le}{\leqslant}
\renewcommand{\preceq}{\preccurlyeq}

\newcommand{\algo}[1]{\mbox{\textsf{#1}}}

\newcommand{\eps}{\varepsilon}
\newcommand{\dallersymbol}{\$} 
\newcommand{\daller}{\dallersymbol} 
\newcommand{\mathdaller}{\mbox{`$\daller$'}} 
\DeclareMathOperator{\polylog}{\mathrm{polylog}}

\newcommand{\kw}[1]{\textbf{#1}}
\newcommand{\sig}[1]{\mathcal{#1}}

\newcommand{\mop}[1]{\mbox{\rm\texttt{#1}}}
\newcommand{\op}[1]{\mathtt{#1}}

\newcommand{\by}{\times}

\renewcommand{\iff}{\mathbin{\,\Leftrightarrow\,}}
\newcommand{\iffdef}{\mathbin{\,\stackrel{\textrm{def}}{\iff}\,}}
\newcommand{\indicator}[1]{\mathbb{I}\kern-0.1em\left[\kern0.1em{#1}\kern0.1em\right]}

\newcommand{\set}[1]{\{\kern0.05em#1\kern0.05em\}}
\newcommand{\sete}[1]{\{\kern0.2em#1\kern0.2em\}}

\newcommand{\rkk}[1]{^{\kern.5pt\textrm{#1}}}

\mathcode`\.=\string"8000
\begingroup \catcode`\.=\active
  \gdef.{\hbox{\hskip1pt\raise-2pt\hbox{$\cdot$\hskip1pt}}}
\endgroup

\usepackage[ruled,vlined,linesnumbered]{algorithm2e} 

\makeatletter
\renewcommand{\@algocf@capt@plain}{above}
\makeatother
\newcommand{\Commentblock}[1]{\hfill$\rhd$\ \textit{#1}}
\newcommand{\Commentblockl}[1]{\kern0.5em$\rhd$\ \textit{#1}\ $\lhd$}
\SetKwComment{Comment}{$\rhd$}{} 
\SetKwInput{KwGlobal}{Global}
\SetKwInput{KwWork}{Working}
\SetArgSty{textrm}
\SetCommentSty{textit}
\newcommand{\iIf}[2]{\textbf{if} {#1} \textbf{then}\hspace{0.125em}{\relax #2}}

\newcommand{\up}{{\mbox{--}}}
\newcommand{\dw}{+}
\newcommand{\down}{\dw}

\newcommand{\inv}{^{-1}}

\newcommand{\lex}{\idrm{lex}}

\newcommand{\pos}{\idrm{pos}}

\newcommand{\idx}{\op{idx}}
\newcommand{\val}{\op{val}}
\newcommand{\Suf}[1][(T)]{\idtt{Suf}{#1}}

\newcommand{\Path}[1][(G)]{\idtt{Path}{#1}}
 
\newcommand{\rext}[1]{\overrightarrow{#1}} 

\newcommand{\str}{\op{str}}

\newcommand{\U}[1][\delta]{\mathcal{U}_{#1}} 
 
\newcommand{\lab}{\op{lab}}
\newcommand{\src}{\op{src}}
\newcommand{\dst}{\op{dst}}
\newcommand{\tail}{\src}
\newcommand{\head}{\dst}

\newcommand{\cano}[1][\delta]{\op{cano}_{#1}}
\newcommand{\precsym}{\op{precsym}}
\newcommand{\Pos}{\op{pos}}
\newcommand{\Rnk}{\op{rnk}}

\newcommand{\CS}[1][\delta]{\sig{CS}_{#1}(G)}
\newcommand{\SP}[1][\delta]{\sig{SP}_{#1}(G)}
\newcommand{\CE}[1][\delta]{\sig{CE}_{#1}(G)}
\newcommand{\RSA}{\mathsf{SA}}
\newcommand{\vobj}[2][\delta]{\bm{#2}_{#1}}
\newcommand{\trivcsuf}[1][\delta]{\vobj[#1]{S}}
\newcommand{\trivsuf}[1][\delta]{\vobj[#1]{S}}
\newcommand{\trivedge}[1][\delta]{\vobj[#1]{f}}

\newcommand{\dom}{\op{dom}}

\newcommand{\less}{<}
\newcommand{\leqp}[1][\delta]{\mathbin{\preceq_{#1}}}
\newcommand{\lessp}[1][\delta]{\mathbin{\prec_{#1}}}
\DeclareMathOperator{\lequp}{\mbox{$\leqp[\up]$}}

\DeclareMathOperator{\leqdw}{\mbox{$\leqp[\down]$}}
\DeclareMathOperator{\lessdw}{\mbox{$\lessp[\down]$}}

\DeclareMathOperator{\leqpos}{\leqp[\pos]}

\DeclareMathOperator{\leqlex}{\leqp[\lex]}
\DeclareMathOperator{\lesslex}{\lessp[\lex]}
\newcommand{\leqedge}{\leq^E}
\newcommand{\lessedge}{<^E}
\newcommand{\leqe}[1][\delta]{{\leqedge_{#1}}}

\newcommand{\leqeout}{\leqe[+]}
\newcommand{\lesse}[1][\delta]{{\lessedge_{#1}}}
\newcommand{\lessein}{{\lesse_{-}}}
\newcommand{\lesseout}{{\lesse_{+}}}

\newcommand{\leqepos}[1][-]{\leq^{E}_{#1,\pos}}
\newcommand{\leqelex}[1][+]{\leq^{E}_{#1,\lex}}

\newcommand{\repr}[1][\delta]{{\mop{repr}_{#1}}}
\newcommand{\reprup}{{\repr[\up]}}
\newcommand{\reprdw}{{\repr[\dw]}}
\newcommand{\valu}[1][]{\reprup}
\newcommand{\shortestup}{\op{shortest}_\up}
\newcommand{\nleaves}{\op{nleaves}}
\newcommand{\isprimary}{\mop{is-primary}}
\newcommand{\isprimaryup}{\isprimary_\up}
\newcommand{\isprimarydw}{\isprimary_\dw}

\newcommand{\Stree}[1][T]{\id{Stree}({#1})}

\newcommand{\CDAWG}{\mathit{CDAWG}}
\newcommand{\CDAWGo}{\CDAWG(T)}
\newcommand{\CDAWGd}{\CDAWG^{-}_{\Pi}(T)}

\newcommand{\SA}{\id{SA}}

\newcommand{\mypreceq}{{\leqdw}}
\newcommand{\myprec}{{\lessdw}}
\newcommand{\qir}[1]{{\widetilde{#1}}}

\newcommand{\GLPF}{\idit{GLPF}}
\newcommand{\GLPFd}{\GLPF_{\scriptsize\mypreceq}}

\newcommand{\V}[1][]{\mathcal{V}_{#1}}
\newcommand{\E}[1][]{\mathcal{E}_{#1}}
\renewcommand{\root}{\id{root}}%
\newcommand{\sink}{\id{sink}}%
\newcommand{\suf}{\id{suf}}%
\newcommand{\N}[1][\delta]{N_{#1}}
\newcommand{\In}{\N[-]}
\newcommand{\Out}{\N[+]}
\newcommand{\EE}[2]{\mbox{$\mathcal{E}_{#2}^{#1}$}}
\newcommand{\EP}[1][\delta]{\EE{\star}{#1}}
\newcommand{\ES}[1][\delta]{\overline{\EE{\star}{#1}}}
\newcommand{\ESX}[1][\delta]{\ES[#1]\cup\set{\trivedge[#1]}}
\newcommand{\T}[1][\delta]{\sig T_{#1}}
\newcommand{\EPU}[1][]{\EP[\up]}
\newcommand{\ESU}[1][]{\ES[\up]}
\newcommand{\EPD}[1][]{\EP[\dw]}
\newcommand{\ESD}[1][]{\ES[\dw]}

\newcommand{\aref}[1]{Algorithm\kern0.25em\ref{#1}} 

\newcommand{\nofbox}[1]{#1}

\begin{document}


\title{
  Optimally Computing Compressed Indexing Arrays 
  Based on the Compact Directed Acyclic Word Graph
}
\titlerunning{Optimally Computing 
  Compressed Indexing Arrays}
\author{
Hiroki Arimura\inst{1}
\and 
Shunsuke Inenaga\inst{2}
\and
Yasuaki Kobayashi\inst{1}
\and 
Yuto Nakashima\inst{2}
\and 
Mizuki Sue\inst{1}
}
\institute{%
  Graduate School of IST, Hokkaido University, Japan
\\\url{{arim,sue}@ist.hokudai.ac.jp}\orcidID{0000-0002-2701-0271} 
\\\url{koba@ist.hokudai.ac.jp}\orcidID{0000-0003-3244-6915} 
\and 
Department of Informatics, Kyushu University, Japan
\\\url{inenaga@inf.kyushu-u.ac.jp}\orcidID{0000-0002-1833-010X} 
\\\url{nakashima.yuto.003@m.kyushu-u.ac.jp}\orcidID{0000-0001-6269-9353} 
}
\authorrunning{H. Arimura et al.}
%
\maketitle              
\thispagestyle{plain}

\begin{abstract}
  In this paper, we present the first study of the computational complexity of converting an automata-based text index structure, called the Compact Directed Acyclic Word Graph (CDAWG), of size $e$ for a text $T$ of length $n$ into other text indexing structures for the same text, suitable for highly repetitive texts: 
the \textit{run-length BWT} of size~$r$, 
the \textit{irreducible PLCP array} of size~$r$, and
the \textit{quasi-irreducible LPF array} of size~$e$,
as well as the \textit{lex-parse} of size~$O(r)$ and
the \textit{LZ77-parse} of size~$z$,
where $r, z \le e$.
As main results, we showed that the above structures can be optimally computed from either the CDAWG for $T$ stored in read-only memory or its self-index version of size $e$ without a text in $O(e)$ worst-case time and words of working space.
To obtain the above results, we devised techniques for enumerating a particular subset of suffixes in the lexicographic and text orders using the forward and backward search on the CDAWG by extending the results by Belazzougui \textit{et al.} in~2015. 

\keywords{
 Highly-repetitive text
  \and suffix tree
  \and longest common prefix
}
\end{abstract}


\section{Introduction}
\label{sec:intro}

\subsubsection{Backgrounds.} 
Compressed indexes for repetitive texts, which can compress a text beyond its entropy bound, have attracted a lot of attention in the last decade
in information retrieval~\cite{navarro2021indexing:ii}.
Among them, the most popular and powerful compressed text indexing structures~\cite{navarro2021indexing:ii} are
the \textit{run-length Burrows-Wheeler transformation} (RLBWT)~\cite{navarro2021indexing:ii}
of size $r$,
the \textit{Lempel-Ziv-parse} (LZ-parse)~\cite{navarro:ochoa:2020approx} of size $z$, and finally the \textit{Compact Directed Acyclic Word Graph} (CDAWG)~\cite{blumer:jacm1987cdawg} of size $e$.
It is known~\cite{navarro2021indexing:ii} that their size parameters $r$, $z$, and $e$ can be much smaller than the information theoretic upperbound of a text for highly-repetitive texts such as
collections of genome sequences and markup texts~\cite{navarro2021indexing:ii}. 
Among these repetition-aware text indexes, we focus on the CDAWG for a text $T$, which is a minimized compacted finite automaton with $e$ transitions for the set of
all suffixes of $T$~\cite{blumer:jacm1987cdawg};
It is 
the edge-labeled DAG obtained
from the suffix tree for $T$ by merging all isomorphic subtrees~\cite{gusfield1997book:stree}, and 
can be constructed from $T$ in linear time and
space~\cite{navarro2021indexing:ii}.
The relationships between the size parameters $r$, $z$, and $e$ of the RLBWT, LZ-parse, and CDAWG has been studied by, e.g.~\cite{radoszewski:rytter2012structure:cdawg:thuemorse,belazzougui2015cpm:composite,kempa:kociumaka2022resolution,bannai2013converting,mantaci2017measuring:burrows:fibonacci,brlek2019burrows:thuemorse};
However, it seems that the actual complexity of conversion the CDAWG into the other structures in sublinear time and space has not been explored yet~\cite{navarro2021indexing:ii}.

\subsubsection{Research goal and main results.} 
In this paper, we study for the first time the conversion problem from the CDAWG for $T$ into the following compressed indexing structures for $T$:
\begin{enumerate}[(i)]
\item the \textit{run-length BWT} (RLBWT)~\cite{navarro2021indexing:ii} of size~$r \le e$; 
  
\item the \textit{irreducible} \textit{permuted longest common prefix} (PLCP)  array~\cite{karkkainen2009permuted} of size~$r$; 
  
\item the \textit{quasi-irreducible} \textit{longest previous factor} (LPF) array~\cite{crochemore2008computing:lpf} of size~$e$ (\cref{sec:prelim}); 

\item the \textit{lex-parse}~\cite{navarro:ochoa:2020approx} with size at most $2r = O(r)$; and
  
\item {LZ-parse}~\cite{navarro:ochoa:2020approx} with size~$z \le e$.
\end{enumerate}

After introducing some notions and techniques,
we present in \cref{sec:algbwt} and \ref{sec:alglpf} 
efficient algorithms for solving the conversion problem from the CDAWG into the aforementioned compressed indexing structures.
We obtain the following results. 

\begin{trivlist}\item[]
  \textbf{Main results (Thm.~\ref{thm:algbwt:complexity}, \ref{thm:alglpf:complexity}, and \ref{thm:algglpf:parse:lex:lz}).}\:  
  For any text $T$ of length $n$ over an integer alphabet $\Sigma$, we can solve the conversion problems from the CDAWG $G$ of size $e$ for $T$ into the above compressed index array structures (i)--(v) for the same text in $O(e)$ worst-case time using $O(e)$ words of working space, where an input $G$ is given in the form of either the CDAWG of size $e$ for $T$ stored in read-only memory, or its self-index version~\cite{belazzougui:cunial2017cpm:representing,takagi:spire2017lcdawg} of size $O(e)$ without a text. 
\end{trivlist}

\subsubsection{Techniques.}
To obtain the above results, we devise in \cref{sec:tech} techniques for enumerating a \textit{canonical subset of suffixes} in the lexicographic and text orders using the \textit{forward and backward DFS} on the CDAWG by extending by~\cite{belazzougui2015cpm:composite}.

\subsubsection{Related Work.}
\label{sec:related}
On the relationships between
parameters $r$, $z$, and $z$ against the text length $n$, 
Belazzougui and Cunial~\cite{belazzougui2015cpm:composite} have shown that $r \le e$ and $z \le e$ hold.
Kempa~\cite{kempa2019optimal} showed that the compressed PLCP and CSA and LZ-parse can be computed in $O(n/\log_\sigma n + r \polylog(n))$ time and space from RLBWT-based index of size $r$ and $T$.  
It is shown by \cite{kempa:kociumaka2022resolution} that the RLBWT of size $r$ for $T$ can be computed from the LZ77-parse of size $z$ for the same text in $r = O(z \polylog(n))$ time and space and $r = O(z \:\log^2 n)$.
Concerning to conversion from the CDAWG $G$ for $T$, we
observe that
$G$ 
can be converted into the LZ78-parse of size $z_{78}\ge z$ in $O(e + z_{78}\log z_{78})$ time and space via an $O(e)$-sized grammar~\cite{belazzougui:cunial2017cpm:representing} on $G$~\cite{bannai2013converting}. 
\textit{Discussions.} 
For some
texts, $e$ can be as small as $r$ or $z$ although $e$ can be polynomially larger than $z$ for other texts~\cite{kempa:kociumaka2022resolution}. 
For the class of Thue-Morse words, \footnote{The $n$-th Thue-Morse word is $\tau_n = \varphi^n(0)$ for the morphism $\varphi(0) = 01$ and $\varphi(1) = 10$. } Radoszewski and Rytter~\cite{radoszewski:rytter2012structure:cdawg:thuemorse} showed that $e = O(\log n)$, while Brlek~\textit{et al}.~\cite[Theorem~2]{brlek2019burrows:thuemorse} showed that $r = \Theta(\log n)$. Hence, for such a class, there is a chance that our $O(e)$-time method can run as fast as other $O(r \polylog(n))$-time methods for some conversion problem. 
On the contrary, Mantaci~\textit{et al}.~\cite{mantaci2017measuring:burrows:fibonacci} showed that $e = \Theta(\log n)$ and $r = O(1)$ for Fibonacci words.


\section{Preliminaries}
\label{sec:prelim}

We prepare the necessary notation and definitions in the following sections. For precise definitions, see the literature~\cite{gusfield1997book:stree,navarro2021indexing:ii} or the full paper~\cite{arimura:full2023thispaper}. 

\subsubsection{Basic definitions and notation.}
For any integers $i\le j$,
the notation $[i..j]$ or $i..j$ denotes the interval $\set{i,i+1, \dots, j}$ of integers, and $[n]$ denotes $\set{1,\dots,n}$.
 For a string $S[1..n] = S[1]\cdots S[n]$ of length $n$ and any $i\le j$, we denote $S[i..j] = S[i]S[i+1]\cdots S[j]$.
 Then, $S[1..j]$, $S[i..j]$, and $S[i..|S|]$ are a \textit{prefix}, a \textit{factor}, and a \textit{suffix} of $S$, resp. The \textit{reversal} of $S$ is $S\inv = S[n]\cdots S[1]$.
 Throughout,
 we assume a string $T[1..n] \in \Sigma^n$, , called a \textit{text}, over an alphabet $\Sigma$ with symbol order $\leq_\Sigma$, which is terminated by the \textit{end-marker} $T[n] = \mathdaller$ such that $\daller \leqlex a$ for $\forall a \in \Sigma$.
$\Suf[(T)] = \set{T_1, \dots, T_n} \subseteq \Sigma^+$ denotes the set of all of $n$ non-empty suffixes of $T$, where $T_p := T[p..n]$ is the $p$-th suffix with position $p$.
 For any suffix $S \in \Sigma^*$ in $\Suf$, we define: 
\begin{enumerate*}[(i)] 
\item $\Pos(S) := n + 1 - |S|$ gives the starting position $S$.
\item $\Rnk(S)$ gives the lexicographic rank of $S$ in $\Suf$.
\end{enumerate*}
$lcp(X, Y)$ denotes the \textit{length of the longest common prefix} of strings $X$ and $Y$.
In what follows, we refer to any suffix as $S$, any factors of $T$ as $X, Y, U, L, P,\ldots$, nodes of a graph as $v, w, \ldots$, and edges as $f, g, \ldots$, which are possibly subscripted. 

\subsubsection{String order and extension.}
A \textit{string order} is any total order $\preceq$ over strings in $\Sigma^*$.
Its \textit{co-relation} $\preceq^\idrm{co}$ is defined by $X \preceq^\idrm{co} Y \iffdef X^{-1} \preceq Y^{-1}$. The order $\preceq$ is said to be \textit{extensible} if
$\forall a \in \Sigma, \forall X, Y \in \Sigma^*, a X \preceq a Y \iff X \preceq Y$, and \textit{co-extensible} if its co-order $\preceq^\idrm{co}$ is extensible.%
\footnote{%
Our extensible order seems slightly different from~\cite{navarro:ochoa:2020approx}, but essentially the same.  
} 
We denote by $\leqlex$ the \textit{lexicographic order} over $\Sigma^*$ extending $\leq_\Sigma$ over $\Sigma$, and
by $\leqpos$ the \textit{text order} defined as:
$X \leqpos Y \iff |X|\ge |Y|$.
Both of $\leqpos$ and $\leqlex$ are extensible~\cite{navarro:ochoa:2020approx}, while $\leqpos$ is co-extensible.
A factor $X$ in $T$ is
\textit{left-maximal} if we can prepend some symbols to $X$
without changing the set of its end-positions in $T$~\cite{belazzougui2015cpm:composite,blumer:jacm1987cdawg}. 

\subsubsection{Compact directed acyclic word graph.}
We assume that the reader is familiar with the suffix tree and the CDAWG~\cite{gusfield1997book:stree,blumer:jacm1987cdawg}.
The \textit{suffix tree}~\cite{gusfield1997book:stree}
for a text $T[1..n]$, denoted by $\Stree$, is the compacted trie for the set $\Suf$ of all suffixes of $T$. 
The CDAWG~\cite{blumer:jacm1987cdawg} for a text $T$, denoted $\CDAWGo$, is an edge-labeled DAG $G = (\V,$ $\E, \suf, \root, \sink)$ 
obtained from $\Stree$ by merging all the isomorphic subtrees, where
$\V$, $\E$, and $\suf$ are sets of
\textit{nodes}, \textit{labeled edges}, and \textit{suffix links},
$\N[-](v)$ and $\N[+](v)$, resp., denotes the sets of incoming and outgoing edges at node $v$.
$\root$ and $\sink\in \V$ are the distinguished nodes with $|\N[-](v)|=0$ and $|\N[+](v)|=0$, resp. 
Each edge $f = (v, X, w)$
goes from node $\tail(f) = v$ to node $\head(f) = w$ with string label $lab(v) = X \in \Sigma^+$.
$\Path[(u,v)]$ denotes the set of all paths from node $u$ to node $v$, whose elements are called \textit{$u$-to-$v$ paths}.
We denote the \textit{size} of $G$ by $e := e_R + e_L$, where 
$e_R := |\E(G)|$ and $e_L := |\suf_G|$. 
The CDAWG can be stored in $O(e)$ words of space by representing each edge label by its length and a pointer to $T$. 
In \cref{fig:stree:cdawg}, we show examples of the suffix tree and the CDAWG for the same text $T = aabaababb$ over $\Sigma = \set{a, b, \daller}$.

\subsubsection{Indexing arrays.}
$\SA, PLCP, LPF \in [n]^n$ and 
$BWT \in \Sigma^n$
denote
the suffix~\cite{navarro2021indexing:ii}, permuted longest common prefix~\cite{karkkainen2009permuted}, longest previous factor~\cite{crochemore2008computing:lpf}, and the BWT arrays for a text $T[1..n]$:
for any rank $k$ and position $p \in [n]$, 
$\SA[k]$ stores the position $p$ of the suffix with rank $k$;
$BWT[k]$ is $T[n]$ if $\SA[k] = 1$ and $T[\SA[k]-1]$ otherwise;
$PLCP[p]$ is $0$ if $p = \SA[1]$ and $lcp(T_p, T_q)$ for $q = \SA[\SA^{-1}[p]-1]$ otherwise; 
$LPF[p]$ is $\max(\set{ lcp(T_p, T_q) \mid T_q \leqpos T_p, q \in [n]  } \cup \set{0})$.
$RLBWT$ is the run-length encoded $BWT$.
The \textit{irreducible PLCP} is obtained from $PLCP$ by sampling such rank-value pairs $(p, PLCP[p])$ that the rank $i = SA^{-1}[p]$ satisfies $BWT[i] \not= BWT[k-1]$.
The \textit{lex-parse}~\cite{navarro:ochoa:2020approx} and \textit{LZ-parse}~\cite{navarro:ochoa:2020approx} of $T$ are obtained from $PLCP$ and $LPF$, resp., as the partition $T = F_1\dots F_u$ of $T$ with $u$ phrases $F_i = T[p_i..p_i+\ell_i-1]$ such that $p_1 = 1$ and 
$\ell_i = \max(L[p_i], 1), \forall i \in [u]$, where $L$ is either $PLCP$ or $LPF$.

\begin{figure}[t]
  \centering  
  \begin{subfigure}[b]{0.28\textwidth}
    \centering
    \nofbox{\includegraphics[width=.75\hsize]{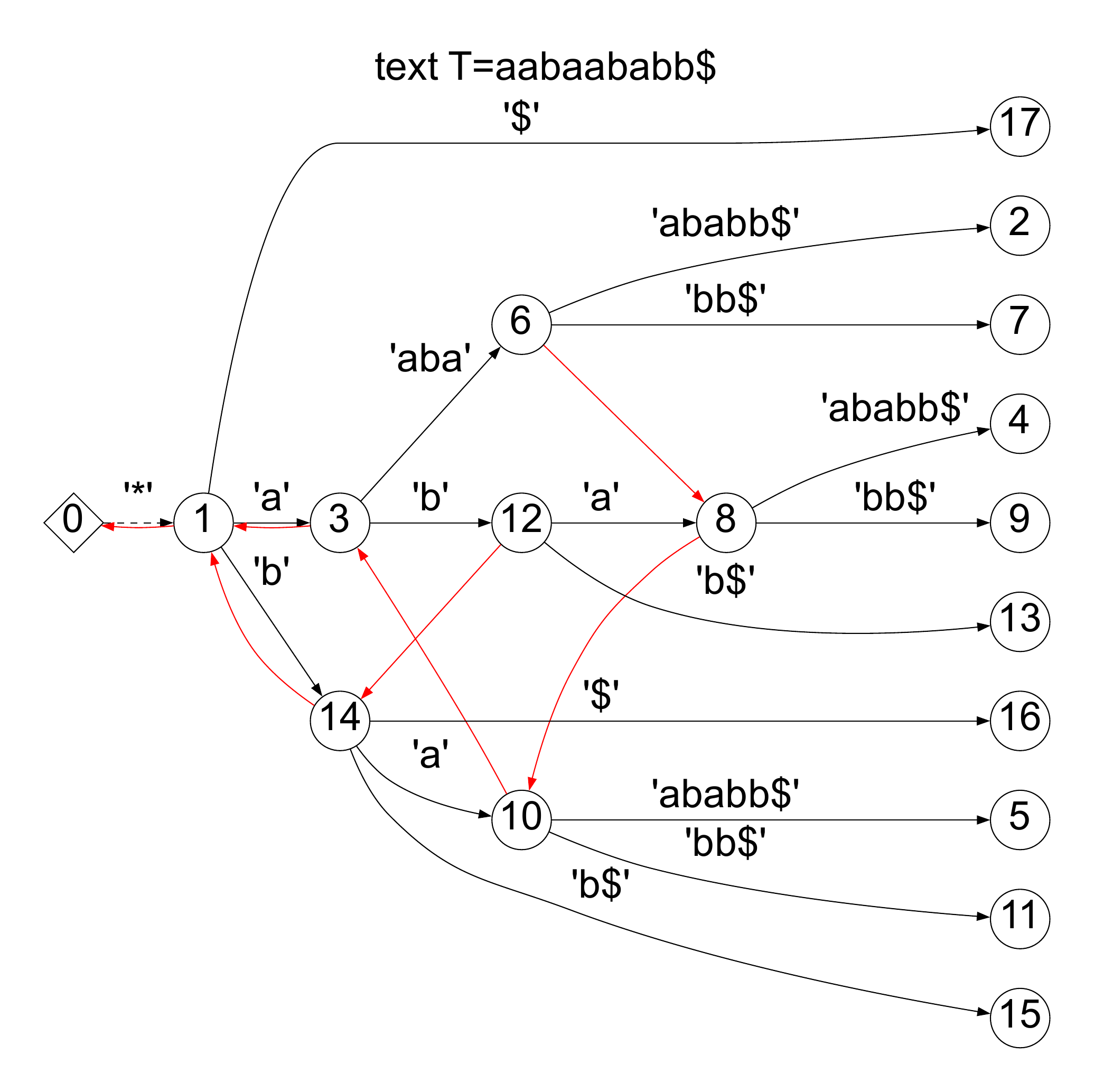}}
  \end{subfigure}
  \begin{subfigure}[b]{0.35\textwidth}
    \centering
    \nofbox{\includegraphics[width=.99\textwidth]{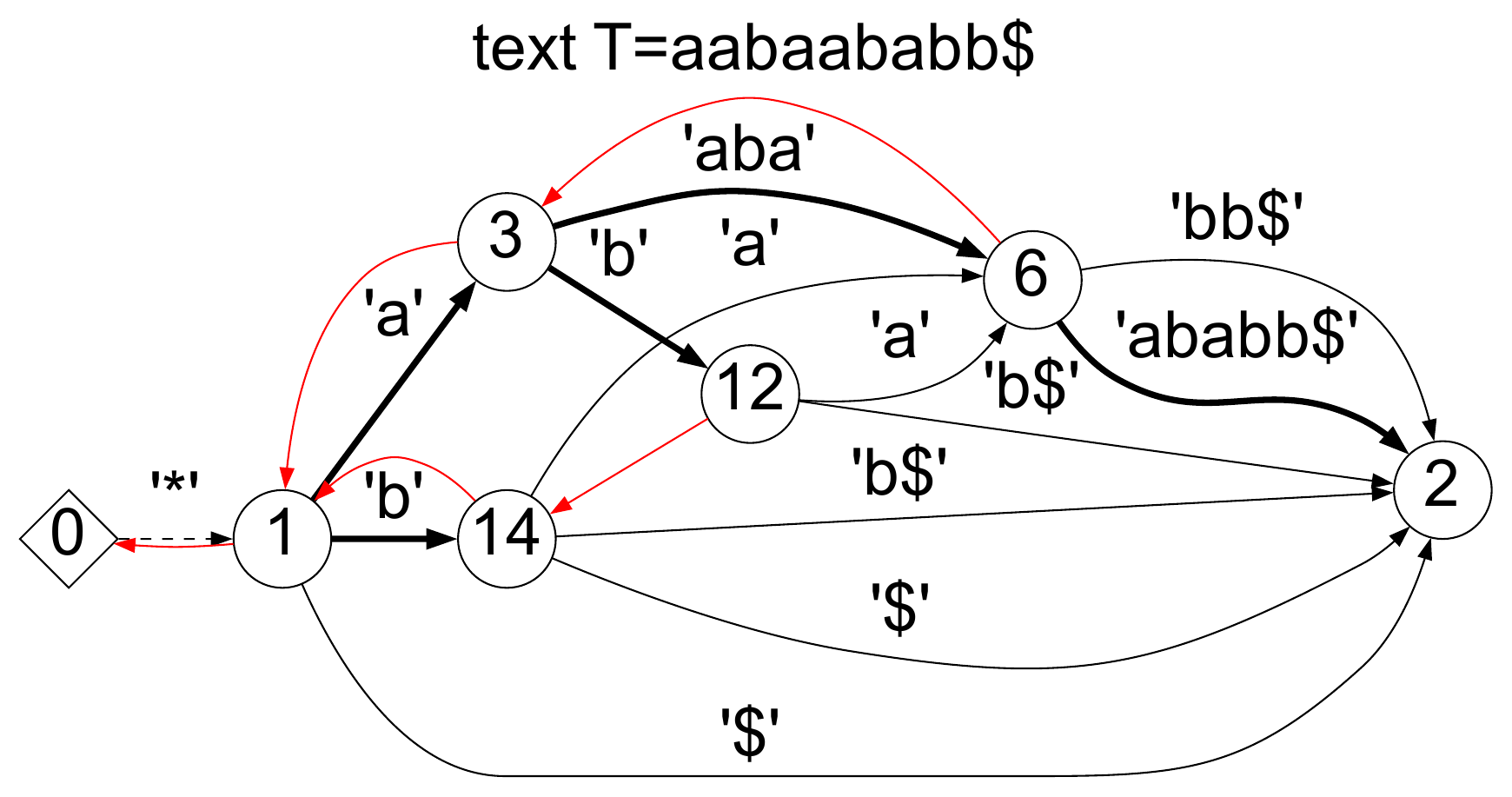}}
  \end{subfigure}
  \begin{subfigure}[b]{0.35\textwidth}
    \centering
    \nofbox{\includegraphics[width=.99\textwidth]{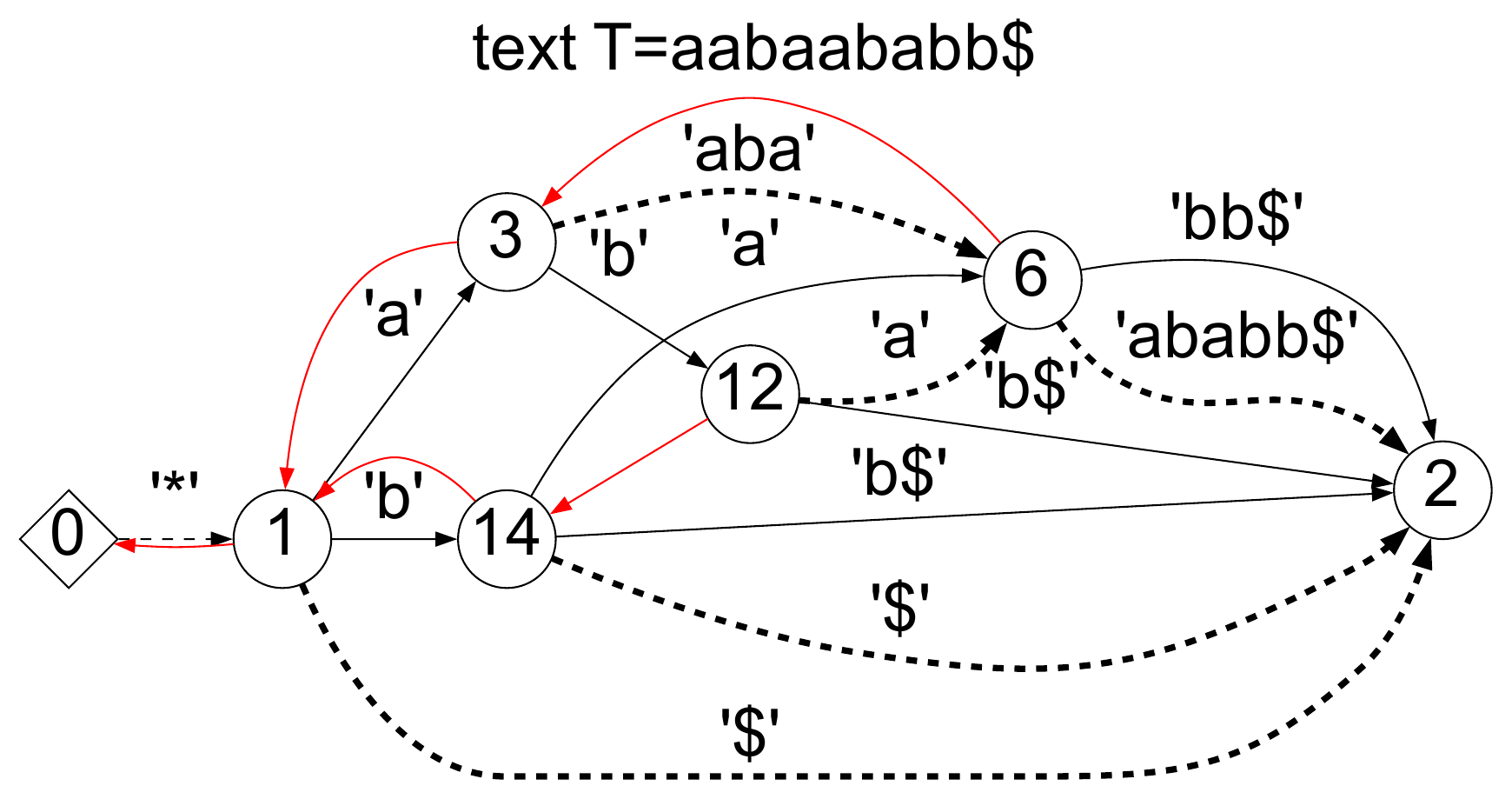}}
 \end{subfigure}
  \caption{Suffix tree (left) and CDAWG (middle, right) for $T = aabaababb\daller$. Thick and dashed lines indicate $(-)$- and $(+)$-primary edges forming the forward and backward search trees $\T[-]$ and $\T[+]$ in \cref{sec:tech:search:tree}, resp. Red lines indicate suffix links. 
  }
  \label{fig:stree:cdawg}
  \label{fig:ordered:cdawg:dfstree}
\end{figure}


\section{Techniques}
\label{sec:tech}
We introduce novel techniques for generating elements of a compressed indexing array using so-called \textit{canonical suffixes} in $\leqlex$ and $\leqpos$ based on the forward and backward DFS on the CDAWG by extending~the results by~\textit{et al.}~\cite{belazzougui2015cpm:composite,belazzougui:cunial2017cpm:representing}. 

\subsubsection{Our approach.}
We employ the one-to-one correspondence between all of $n$ root-to-sink paths in $\Path$ and all of $n$ non-empty suffixes of $T$ in $\Suf \subseteq \Sigma^*$.
Due to the determinism of $\CDAWGo$ as a DFA~\cite{blumer:jacm1987cdawg}, we can interchangeably use a path $\pi = (f_1, \dots, f_k)$ in $\Path$ and a factoring $S = X_1\dots X_k = \str(\pi)$ of a suffix $S$ in $\Suf$, where $X_i = \lab(f_i)$ for all $i \in [k]$.%
\footnote{
The correspondence between $\Path$ and $\Suf$  can be extended to that between all root-to-node paths in $G$ and all \textit{right-maximal factors}~\cite{blumer:jacm1987cdawg} of $T$, but not used here. }

A basic idea of our approach for computing a sparse indexing array $\widetilde{A}: Dom \to Range$ with domain $Dom$ is to represent $\widetilde{A}$ as the graph 
\begin{align}
  \widetilde{A} &= \sete{ (\op{idx}(S), \op{val}(S)) \mid S \in \sig S}
  \subseteq Dom\by Range
\end{align}
of array $\widetilde{A}$ with the set of indexes $\sete{ \op{idx}(S) \mid S \in \sig S}$ using a combination of 
\begin{enumerate}[(i)]
\item a subset $\sig S\subseteq \Path$ of root-to-sink paths, 
\item a mapping $\idx: \sig S \to Dom$ that assigns the index, and
\item a mapping $\val: \sig S \to Range$ that assigns the value.  
\end{enumerate}

For actual computation of $\widetilde{A}$ on the CDAWG $G$ for $T$, we make the DFS based on a pair $\Pi = (\leqp[\up], \leqp[\dw]) = (\leqpos, \leqlex)$ of ordering over paths of $G$. Two ordering have different roles. For example, if we want to compute the run-length BWT for $T$ (\cref{sec:algbwt}),
the set of primary edges (in the sense of \cite{blumer:jacm1987cdawg}) w.r.t.~the first ordering $\leqpos$ defines a spanning tree $\T[]$ over $G$ from the root, whereas the second ordering $\leqlex$ specifies the order of traversal. Finally, the set of secondary edges w.r.t.~$\leqpos$ provides a collection $\sig C$ of target values to seek in the DFS. Actually, we can extract an equal-letter run from each secondary edge in constant time.

On the contrary, if we want to compute the sparse version of LPF array (\cref{sec:alglpf}), we make the backward DFS of $G$ from the sink based on the pair $\Pi = (\leqpos, \leqpos)$. Then, the set of primary edges w.r.t.~the second ordering $\leqpos$ defines a spanning tree, the first ordering $\leqpos$ specifies the text order, and the set of secondary edges w.r.t.~the first ordering $\leqpos$ provides a collection of target values, which are the LCP values of neighboring suffixes. The PLCP array can be computed in a similar manner, but with the pair $\Pi = (\leqpos, \leqlex)$. 



\subsubsection{Ordered CDAWGs.}
We assume a pair $\Pi = (\leqp[\up], \leqp[\dw])$ of co-extensible and extensible string orders, where $\leqp[\up]$ and $\leqp[\dw]$ are called \textit{upper} and \textit{lower path orders}, resp.
The ordered CDAWG for $T$ under a pair $\Pi = (\leqp[\up], \leqp[\dw])$ of path orders, denoted by $G = \CDAWG(T; \leqp[\up], \leqp[\dw])$, is the CDAWG $G$ for $T$ whose incoming and outgoing edges are ordered by edge orderings $(\leqe[-], \leqe[+])$ compatible with path orders defined as follows.  
In what follows, for any node $v$, each sign $\delta \in \set{-,+}$ indicates the \textit{side} of $G$, where 
the prefix $(-)$- reads ``upper,'' while $(+)$- reads ``upper'' in what follows.

\subsubsection{Upper/lower sets and their representatives.}
Consider the sets
$\U[-](v) := \sete{ \str(\pi) \mid \pi \in \Path[(\root(G), v)] }$ and $\U[+](v) := \sete{ \str(\pi) \mid \pi \in \Path[(v, \sink(G))] } \subseteq \Sigma^*$
of prefixes and suffixes of all root-to-sink paths.%
\footnote{
The set $\U[-](v)$ has appeared as the equivalence class $[\rext{X}]_R$ of all factors with the same end positions in~\cite{blumer:jacm1987cdawg}, while $\U[+](v)$ was recently introduced by~\cite{belazzougui2015cpm:composite,belazzougui:cunial2017cpm:representing}. Indeed, $\U[-](v)$ encodes the node $v$ itself, while $\U[+](v)$ encodes all end positions of such factors~\cite{blumer:jacm1987cdawg}.
} 
Each member of $\U[-]$ and $\U[+]$ are called \textit{upper} and \textit{lower paths}, resp.
For any $\delta \in \set{-,+}$, we define the \textit{$\delta$-representative} of the set $\U$ by the smallest element $\repr(v)$ of $\U(v)$ under $\leqp$, i.e., $\repr(v) := \min_{\leqp} \U(v)$;
For example, $\reprup(v)$ is the longest strings in $\U[-](v)$ and $\reprdw(v)$ is the lex-first string in $\U[+](v)$ under $\Pi (\leqpos, \leqlex)$. 

\begin{remark}
  For any $\delta \in \set{-, +}$, for any $P \in \sig U_\delta(v)$, $P = \repr(v)$, if and only if $P$ consists of $\EP[\delta]$-edges only.
  Furthermore, any factor $P$ is left-maximal if and only if $P = \reprup(v)$ for some node $v$ under $\leqp[\up] = (\leqpos)$. 
\end{remark}


\subsubsection{Compatible Edge orderings.}
Under the
pair  $\Pi^\pos_\lex = (\leqpos, \leqlex)$ of path orderings, we define the pair $\Gamma^\pos_\lex = (\leqepos[-], \leqelex[+])$ of \textit{edge orderings} by 
\vspace{-0.5\baselineskip}
\begin{align*}
  f_1 \leqepos[-] f_2 
  &\iffdef |\repr[-](v_1)| + |X_1| \geq |\repr[-](v_2)| + |X_2|
  \\f_1 \leqelex[+] f_2 
  &\iffdef \lab(f_1)[1] \less_\Sigma \lab(f_2)[1], 
\end{align*}
$f_i = (v_i, X_i, w_i) \in \E$ be an edge for $i=1,2$.  
Under $\Pi^\pos_\pos = (\leqpos, \leqpos)$,
we define the pair $\Gamma^\pos_\pos = (\leqepos[-], \leqepos[+])$ of edge orderings, where
\begin{math}
f_1 \leqepos[+] f_2 
  \iffdef |X_1| + |\repr[-](w_1)| \geq |X_2| + |\repr[-](w_2)|. 
\end{math}

\begin{toappendix}
  \begin{remark}
Let us denote by $\lesse$ and $\lessp$ the strict orders of $\leqe$ and $\leqp$, resp. 
Under $\Pi^\pos_\lex = (\leqpos, \leqlex)$, the edge ordering $\leqe$ with $\delta \in \set{-,+}$ is \textit{compatible to} the path ordering $\leqp$ w.r.t.~$\Pi^\pos_\lex$ in the following sense
\newcommand{\ucdot}{}
\begin{align}
    f_1 \lesse[-] f_2
    &\iff
    \forall U_1\!\in\U[-](v_1), U_2\!\in\U[-](v_2), 
    U_1\ucdot X_1 \lessp[-] U_1\ucdot X_2,
    \forall f_1, f_2\!\in \N[-](v)
    \label{eq:compatible:up}
    \\
    f_1 \lesse[+] f_2
    &\iff
    \forall U_1\!\in\U[+](w_1), U_2\!\in\U[+](w_2), 
    X_1\ucdot U_1 \lessp[+] X_2\ucdot U_2, 
    \forall f_1, f_2\!\in \N[+](v) 
    \label{eq:compatible:dw}
\end{align}
where \cref{eq:compatible:up} follows from~\cite[Theorem~1]{belazzougui2015cpm:composite} and \cref{eq:compatible:dw} follows the determinism of the CDAWG. 
For $\Pi^\pos_\pos$, $\lessp[-]$ satisfies the above property, while 
$\lessp[+]$ does not.
\end{remark}
\end{toappendix}

\subsubsection{Classification of edges.}
\label{sec:tech:edge:class}
We classify edges in $\N(v)$ using the representative $\repr(v)$ under $\leqe$ as follows. 
For $\delta \in \set{-, +}$, $\delta$-edge $f \in N_{\delta}(v)$ is said to be \textit{$\delta$-primary} if $\repr(v)$ goes through $f$. We denote by $\EP[\delta]$ the set of \textit{all $\delta$-primary} edges, and by $\ES[\delta] := \E - \EP[\delta]$ the set of all \textit{$\delta$-secondary} edges.
The same edge can be both $(-)$-primary and $(+)$-secondary, and \textit{vice versa}.
We remark that
it gives the partition $\E = \EP\uplus\ES$ and equivalence $|\EP[-]| = |\EP[+]|$ and $|\ES[-]| = |\ES[+]|$.%
\footnote{
This is because any nonempty set $\N(v)$ has at least one $\delta$-primary edge.
}
Any suffix $S\in \Suf$ is
\textit{$\delta$-trivial} if it consists only of $\EP$-edges, and \textit{$\delta$-nontrivial} otherwise, where 
the $\delta$-trivial one is unique and denoted by $\trivcsuf$. 
We assume that $\trivcsuf$ has an imaginary edge $\trivedge$
\footnote{
We assume to add imaginary edges $\trivedge[-]$  and $\trivedge[+]$, resp., which are attached above $\root$ and below $\sink$, into the sets $\ES[-]$ and $\ES[+]$.
} 
at the \textit{bottom} if $\delta=(-)$ and at the \textit{top} if $\delta=(+)$.

\subsubsection{Preprocessing.}
We observe that preprocessing of $G = \CDAWGo$ for the information  
necessary in Sec.~\ref{sec:algbwt} and \ref{sec:alglpf} can be efficiently done as follows~\cite{belazzougui2015cpm:composite,belazzougui:cunial2017cpm:representing}.

\begin{lemmarep}[preprocessing]
  \label{lem:preprocessing}
  Under a pair $\leqp[-] = \leqpos$ and $\leqdw$ of co-extensible and extensible string orders, 
  we can preprocess $\CDAWGo$ in $O(e)$ worst-case time and words of space
  to support the following operations in $O(1)$ time
  for $\forall v \in \V$: 
  \begin{enumerate}[(i)]
  \item $|\repr(v)|$ returns the length $\ell \in 0..n$ of $\repr(v)$,
    and $\isprimary_{\delta}(f) \in \set{0,1}$ indicates if
    $f$ is $\delta$-primary for $\delta \in \set{-, +}$ 
    under
    $\leqp[+] \in\set{\leqlex,\leqpos}$;
    
  \item $|\shortestup(v)|$ returns the length $|U|\in [n]$ and $\op{fstsym\-shortest}(v)$ returns the start symbol $U[1]\in \Sigma$ of
    the shortest string $U$ in $\U[-](v)$.

  \item
    $\nleaves(v) \in \nat$ returns the number $|\U[+](v)|$ of lower paths below $v$. 
  \end{enumerate}
\end{lemmarep}

As usual, edges of the CDAWG are assumed to be sorted according to $\leqelex[+]$ and $\leqepos[-]$.
If needed, it can be done in $O(e)$ time and space; the sorting with $\leqelex[+]$ is done by transposing an incident matrix between nodes and edges, while the sorting with $\leqepos[-]$ is done by traversing either the suffix links with a read-only text or suffix links of type-ii nodes~\cite{takagi:spire2017lcdawg} in a self-index.

\begin{toappendix}

  We give the proof of \cref{lem:preprocessing} for preprocessing based on three recursive procedures
\algo{PrepUpward} in \aref{algo:mark:up} and 
\algo{PrepDownward} in \aref{algo:mark:down}. 

\begin{algorithm}[t]
\caption{
  The procedure for precomputing tables
  $\isprimaryup(f)$,
  $|\reprup(v)|$, and 
  $|\shortestup(v)|$
  under $\lequp = (\leqpos)$, 
  invoked with $v = \sink(G)$, and initialized with 
  $\op{Visited}[v]= 0$,
  $\isprimaryup(f)= 0$ for all $v \in \V, f \in \E$.
}\label{algo:mark:up}
\kw{Procedure} \algo{PrepUpward}$(v)$\;
\Begin{
    \uIf{$\op{cache}[v] \not= \bot$}{
      \Return\;
    }
    \uElseIf{$\N[-](v) = \emptyset$}{
      Record $|\reprup(v)| \gets 0$, $|\shortestup(v)| \gets 0$\; 
      \Return\;
    }
    \Else{
      $len_* \gets 0$;\:  
      $len'_* \gets 0$;\:  
      $f_* \gets \bot_{edge}$
      \Comment*{$f \prec \bot_\idrm{edge}$ for all $f\in E$} 
      \For{\kw{each} $f = (w, X, v) \in \N[-](v)$
      } {
        \algo{PrepUpward}$(w)$\;
        $len \gets |\reprup(w)|$; 
        $len' \gets |\shortestup(w)|$\; 
        \iIf{ $len + |X| > len_*$ } {
          $len_* \gets len + |X|$; \: 
          $f_* \gets f$\; 
        }
        \iIf{ $len' + |X| < len'_*$ } {
          $len'_* \gets len' + |X|$\;
        }
      }
      \iIf{$f_* \not= \bot$}{
        $\isprimaryup(f_*) \gets 1$\;
      }
      Record $|\reprup(v)| \gets len_*$ and 
      $|\shortestup(v)| \gets len'_*$\;
      $\op{cache}[v] \gets 1$\; 
      \Return\;
      }
  }
\end{algorithm}

\begin{lemma}[computation of $|\reprup(v)|$ and $|\shortestup(v)|$]
  \label{lem:algo:mark:up}
  The procedure \algo{PrepUpward} in \aref{algo:mark:up}
  preprocesses $\CDAWGo$ in $O(e)$ worst-case time and words of space
  to support the operations 
  $|\reprup(v)|$, 
    $\isprimaryup(f)$, and 
  $|\shortestup(v)|$
  to be answered in $O(1)$ time
  for all $v \in \V$ and $f \in \E$. 
\end{lemma}
  
\begin{proof}
  Firstly, consider the computation of the tables for 
  $|\reprup(v)|$ and $|\shortestup(v)|$ for $\delta = (-)$, 
  which is done by the procedure \algo{PrepUpward} in \aref{algo:mark:up} using the reverse DFS of $\CDAWGo$ starting from $\sink(G)$ and following incoming edges upwards. 
  \algo{PrepUpward} recursively fills the tables using the following recurrence:
  \begin{enumerate}[(a)]
  \item If $v = \root(G)$, obviously $|\reprup(v)| = 0$ and $|\shortestup(v)| = 0$; 
  \item Otherwise, $\N[-](v)\not= \emptyset$, and we have
  \begin{enumerate*}[(i)]
   \item $|\reprup(v)| := \max \sete{ |\reprup(w)| + |X|\mid$ $(w, X, v) \in \N[-](v) }$;  
   \item $|\shortestup(v)| := \min \sete{ |\reprup(w)| + |X| \mid$ $ (w, X, v)$ $\in$ $ \N[-](v) }$.
  \end{enumerate*}
  \end{enumerate}
  Then, for each $f = (v, X, w) \in \N[-](v)$, the table $\isprimaryup(f)$ can be computed as:
  $\isprimaryup(f)$ is $1$ if $|\reprup(v)| + X$ is maximum over all $f \in \N[-](v)$, and~$0$ otherwise.  
\qed
\end{proof}

Next, we consider the computation of the tables for $|\reprdw(v)|$ and $\isprimaryup(f)$ for $\delta = (+)$,

\begin{algorithm}[t]
\caption{
The procedure for precomputing tables
$\isprimarydw(f)$,
$|\reprdw(v)|$, and 
$\op{info\_repr}_+(v)$ 
under $\leqdw$, 
invoked with
$v = \root(G)$
and initialized with 
$\op{cache}[v]\gets 0$ and 
$\isprimarydw(f)\gets 0$, 
for all $v \in \V, f \in \E$.
}\label{algo:mark:down}
\textbf{Input}: A node $v$\; 
\textbf{Working variables}: The information $D = \op{repr\_info}(v)$ of the representative $\reprdw(v)$ for $\U[+](v)$ under $\leqdw$ with length $len = |\reprdw(v)|$
\footnote{
We remark that the concatenation $(X\cdot D)$ and test $(X\cdot D) \:\lessdw\: D_*$ can be implemented to execute in $O(1)$ time for $\leqdw \in \set{\leqpos, \leqlex}$ under an appropriate representation of $D$. 
}\; 
\kw{Procedure} \algo{PrepDownward}$(v; \preceq)$\;
\Begin{
    \iIf{$\op{cache}[v] \not= \bot$}{
      \Return\;
    }
    \uElseIf{$\N[+](v) = \emptyset$}{
      $\op{repr\_info}_+(v)| \gets \eps$;\:
      $|\reprdw(v)| \gets 0$\;
      \Return\Comment*{the empty downward path $\eps$}
    }
    \Else{
      $D_* \gets \bot$;\: 
      $len_* \gets 0$;\:
      $f_* \gets \bot$%
      \Comment*{
        $\pi \leqdw \bot\: (\forall \pi\in \U[+](v))$} 
      \For{\kw{each} $f = (v, X, w) \in \N[+](v)$} {
        \algo{PrepDownward}$(w; \preceq)$
        \Comment*{Recursive call}
        $D \gets \op{repr\_info}_+(w)$; \: 
        $len \gets |\reprdw(w)|$\; 
        \If{ $(X\cdot D) \:\lessdw\: D_*$ } {
          $D_* \gets (X\cdot D)$;\: 
          $len_* \gets |X| + len$;\:
          $f_* \gets f$\; 
        }
      }
      \iIf{$f_* \not= \bot$}{
        $\isprimarydw(f_*) \gets 1$\;
      }
      $\op{repr\_info}_+(v)| \gets D_*$;\:
      $|\reprdw(v)| \gets len_*$\;
      \Return\; 
    }
  }
\end{algorithm}

\begin{lemma}[computation of $|\reprdw(v)|$ and $\isprimaryup(f)$]
  \label{lem:algo:mark:down}
  The procedure \algo{PrepDownward} in \aref{algo:mark:down}
  preprocesses $\CDAWGo$ in $O(e)$ worst-case time and words of space
  to support the operations 
  $|\reprdw(v)|$ and
  $\isprimarydw(f)$ 
  to be answered in $O(1)$ time
  for $\forall v \in \V$
  and $\forall f \in \E$. 
\end{lemma}

\begin{proof}
The computation of the operations can be done by \algo{PrepDownward} in \aref{algo:mark:down}
  using the DFS of $\CDAWGo$ starting from $\root(G)$ and following outgoing edges downwards. 
  This can be done by the following recurrence:
  \begin{enumerate}[(a)]
  \item If $v = \sink(G)$, we obviously have $|\reprdw(v)| = 0$; 
  \item Otherwise, $\N[+](v)\not= \emptyset$. Then, in general, we have that
    $|\reprdw(v)|$ is the length of the representative 
    $\reprdw(v)$, which is the minimum of $(X \cdot \reprdw(w))$ over all outgoing edges $f = (v, X, w) \in \N[+](v)$ under $\leqdw$. 
  \end{enumerate}
  Based on the above recurrence, we can compute the table $|\reprdw(v)|$. Once table for $|\reprdw(v)|$ is computed, the predicate $\isprimarydw(f)$ for each $f = (v, X, w) \in \N[-](v)$ can be computed generally as:  $\isprimarydw(f)$ is $1$ if $(X\cdot\reprdw(w))$ is maximum over all $f \in \N[-](v)$ under $\leqdw$, and~$0$ otherwise. 
  For the time complexity of \algo{PrepDownward}, we see that it visits each edge exactly once by using the mark $\op{cache}$ in $O(e)$ space. The bottleneck here is the test for the ordering $X_1 D_1 \leqdw X_2 D_2$ (*) between a pair of concatenated paths $X_iD_i$, where $f = (v, X_i, w_i)$ and $D_i \in \U[+](w_i)$ for every $i=1,2$. In general, this may take more than constant time.
  Fortunately,
we can show the claim that 
for $\leqdw \in \set{\leqpos, \leqlex}$, the test $X_1 D_1 \leqdw X_2 D_2$ (*) can be done in $O(1)$ time using $O(1)$ space to hold the paths
paths $X_iD_i$.
This can be shown  by considering the following two cases:
\begin{enumerate*}[(1)]
\item If $\leqdw = (\leqpos)$, we observe that $X_1 D_1 \leqdw X_2 D_2$ iff $|X_1|+|D_1| \ge |X_2|+|D_2|$. Thus, we can decide $\leqpos$ in $O(1)$ time by holding the length $|X_iD_i|\ge 1$ of a path in $O(1)$ space.
\item If $\leqdw = (\leqlex)$, we observe that $X_1 D_1 \leqdw X_2 D_2$ iff $X_1[1] \leq_\Sigma X_2[1]$. Thus, we can decide $\leqlex$ in $O(1)$ time by holding the first symbol $X_i[1] \in \Sigma$ of a path in $O(1)$ space. 
\end{enumerate*}
Hence, the algorithm can precompute the tables in $O(e)$ time and space. 
\qed
\end{proof}

Now, we have the proof of \cref{lem:preprocessing}.

\begin{trivlist}\item[]
  \textit{Proof of \cref{lem:preprocessing}}
  It is shown in \cite{belazzougui2015cpm:composite} that the operation $\nleaves$ can be answered in $O(1)$ time after preprocessing of $G$ in $O(e)$ time and space. Therefore, the lemma immediately follows from 
  \cref{lem:algo:mark:up} and 
  \cref{lem:algo:mark:down}. 
  \qed 
\end{trivlist}
\end{toappendix}


\subsubsection{Canonical suffixes.}
Next, we introduce a set $\CS$ of canonical suffixes, a set $\SP$ of search paths, 
and a mapping $\cano: \ES\cup\set{\trivedge} \to \CS$. 

\begin{definition}[canonical suffix and search path]\rm
\label{def:canonical:suffix:search:path}  
For $\delta \in \set{-,+}$,
we define
a \textit{$\delta$-canonical suffix} $S$,
its \textit{$\delta$-certificate} $f$, and 
its \textit{$\delta$-search path} $P$, where 
$\pi = (f_1, \dots, f_\ell)$, $\ell \ge 1$ is any path in $\Path$ spelling a suffix $S$ in $\Suf$: 
\begin{enumerate}[(a)]
\item In the case that $S$ is trivial. Then, $S = \trivsuf = \repr(end_\delta)$ and it is always canonical, where $end_- = \sink$ and  $end_+ = \root$. 
  Then, the $\delta$-certificate is $f = \trivedge$, and 
  the $\delta$-search path for $\trivedge$ is $P = \trivsuf$ itself.
  Let $\cano(f) = \trivsuf$. 

\item In the case that $S$ is non-trivial.
$S$ is $\delta$-canonical if it has a \textit{$\delta$-canonical factoring} defined below with some index $k \in [\ell]$ of an edge in $\pi$: 
  \begin{enumerate}[(i)]
  \item $f_\delta = f_k \in \ES$, and moreover, 
    if $\delta = (-)$ then $f_k$ is the highest $\ES[-]$-edge in $S$,
    and if $\delta = (-)$ then $f_k$ is the lowest $\ES[+]$-edge in $S$; 
  \item the upper path $U_\delta = (f_1,\dots, f_{k-1})$ consists only of $\EP[-]$-edges;  
  \item the lower path $L_\delta = (f_{k+1},\dots, f_\ell)$ consists only of $\EP[+]$-edges;  
  \item
  the factoring is $S = \str(U_\delta)\cdot X_\delta \cdot \str(D_\delta)$, where 
  $X_\delta$ is the label of the edge 
  $f_k = (v_\delta, X_\delta, w_\delta)$. 
  \end{enumerate}
Then, the \textit{$\delta$-certificate} is $f = f_k$, and the \textit{$\delta$-search path} for $f$ is the path $P = U \cdot X$ for $\delta = (-)$ and the path $P = X\cdot U$ for $\delta = (+)$, where $X = \lab(f)$. 
  Let $\cano(f) = S$. 
  \hfill $\diamond$
\end{enumerate}
\end{definition}

In what follows, we denote by 
$\CS \subseteq \Suf$ and $\SP \subseteq (\E)^*$ the \textit{set of all $\delta$-canonical suffixes} of $T$ and the \textit{set of all $\delta$-search paths} of $G$, reps., under $\Pi$.  
We remark that any $\delta$-canonical suffix $S \in \CS$ can be recovered by its $\delta$-certificate edge $f$ via $\cano$, and thus the mapping $\cano$ is well-defined. 

\begin{lemma}\label{lem:canonical:unique}
  For any $\delta\in\set{-,+}$ and any $\delta$-canonical suffix $S$, its $\delta$-canonical factoring and $\delta$-certificate $f_\delta$ are unique. Consequently, the mapping $\cano$ is a bijection between $\ES\cup\set{\trivedge}$ and $\CS$. 
\end{lemma}

\begin{lemmarep}[properties of canonical suffixes]
  \label{lem:canonical:factor}
  Under any pair $\Pi$ of path ordering, any suffix $S$ in $\Suf$ satisfies conditions (1) and (2) below: 
  \begin{enumerate}[(1)]
  \item $S$ has $(-)$-canonical factoring if and only if
    $S$ has $(+)$-canonical factoring. 
    
  \item Let $S$ be any canonical suffix with the associated path $\pi = (f_1, \dots, f_\ell) \in \Path$ spelling $S$, and let $f_{k_-}, f_{k_+}$ be the indexes of the $(-)$- and $(+)$-certificate $k_-, k_+$  in $\pi$,  resp., then $1\le k_+ \le k_- \le \ell$ holds.
  \end{enumerate}
\end{lemmarep}

\begin{proof}
  (1) First, we show the direction that $(-)$-canonical factoring of $S$ implies the $(+)$-canonical factoring of $S$. 
  Let $\pi = (f_1,\dots,f_\ell)$ be the path spelling $S = \str(\pi)$.  Suppose that $S$ has a $(-)$-canonical factoring and a $(-)$-certificate $f_k$ with $k$ such that $U_- = (f_1,\dots,f_{k - 1}) \in (\EP[-])^*$, $f_k \in \ES[-]$, and $L_- = (f_{k + 1}, \dots, f_\ell) \in (\EP[+])^*$. Consider the prefix $U_-\cdot f_k = (f_1,\dots,f_{k - 1}, f_k)$ of $\pi$.
  Then, we can always take the possibly empty, maximal prefix $L_+ := (f_i,\dots, f_k)$ of $U_-\cdot f_k$ consisting of $\EP[+]$-edges only, where $1\le i\le k$.
  There are two cases below:
  (1.a) If $i > 1$, we have $f_{m} \in \ES[+]$ for $m := i-1$ by the maximality of $L_+$. Then, the possibly empty prefix $U_+ := (f_1,\dots, f_{i-2})$ always exists and is a prefix of $U_-$. Therefore, $U_+$ must consist of $\EP[-]$-edges only since so is $U_-$.
  Putting $k_- := k$ and $k_+ := m$, we showed the claim.
  (1.b) Otherwise, $i = 1$. In this case, $U_-\cdot f_k$ consists of $\EP[+]$-edges only, and then, so is $S$. Since this means that $S$ is a $(+)$-trivial, $S$ contains a virtual $\ES[+]$-edge on its top with index $0$. Thus, it has $(+)$-canonical factoring with $k_+ := 0$. Combining the above arguments, we have claim (1).
  By construction above, it immediately follows that $k_+ = i-1 \le i \le k = k_-$. This shows claim (2).
  By symmetry, we can similarly show the opposite direction with the exception that we put $k_- = \ell + 1$ if $S$ is $(-)$-trivial. Hence, the lemma is proved. \qed 
\end{proof}

By \cref{lem:canonical:factor}, $\CS[-] = \CS[+]$ holds.
Thus, we denote the set by $\CS[] := \CS[-] = \CS[+]$, and simply call its members \textit{canonical suffixes} of $G$.

\subsubsection{Forward and backward DFSs using search paths.}
\label{sec:tech:search:tree}
Recall that in our approach, we encode a target indexing array, say $C$, with a subset $\CS[]$ of canonical suffixes under some path ordering $\Pi = (\leqp[-], \leqp[+])$ and an appropriately chosen pair $\varphi = (\idx, \val)$ of mappings over $\CS[]$ as the image of $\CS[]$ by $\varphi$. Thus, the remaining task is to generate all index-value pairs $\varphi(S) = (\idx(S), \val(S))$ by enumerating all $S \in \CS[]$ in the appropriate index order $\leqlex$ or $\leqpos$.
To do this, we use the forward and backward DFSs using search paths of $\SP$ as follows.

Consider the directed graph $\T[-]$ (resp.~$\T[+]$) obtained from $\SP[-]$ (resp.~$\SP[+]$) by merging common prefixes (resp.~suffixes). Then, we can easily see that  (i) $\T[-]$ is connected at the root (resp.~so is $\T[+]$ at the sink), (ii) $\T[-]$ is spanning over $\ES[-]$ (resp.~so is $\T[-]$ over $\ES[+]$). However, the graph $\T$ may contain the same edge more than once. The next lemma states that it is not the case for $\SP$. In \cref{fig:ordered:cdawg:dfstree}, we show examples of the forward and backward search trees $\T[-]$ and $\T[+]$. 

\begin{lemma}\label{lem:sp:freeness}
  Let $(\EP[-], \EP[+])$ be any pair of partitions of $\E$. Then, 
  \begin{enumerate*}[(1)]
  \item the set $\SP[-]$ is prefix-free, and 
  \item the set $\SP[+]$ is suffix-free. 
  \end{enumerate*}
\end{lemma}

\begin{proof}
  (1) Since $\SP[-]\subseteq (\EP[-])^*\cdot \ES[-]$, for any distinct search paths $X, Y$ in $\SP[-]$, $X$ cannot be a proper prefix of $Y$. By symmetry, we can show (2).
  \qed 
\end{proof}

\begin{proposition}\label{lem:sp:spanning:trees}
  $\T[-]$ is a forward spanning tree for $\ESX[-]$ rooted at $\root$, while $\T[+]$ is a backward spanning tree for $\ESX[+]$ rooted at $\sink$. Moreover, $\T[-]$ and $\T[+]$ have at most $e$ edges.%
\end{proposition}

\begin{proof}
    Since properties (i) and (ii) were proved, we show that (iii) any distinct paths of $\SP$ do not share the same edge in common. 
    Precisely speaking, $\T$ is a rooted tree whose nodes are edges of $G$. Note that if we embed it into $G$, it may form a DAG in general. However, we can show the claim (iii) from \cref{lem:sp:freeness}. From claim (iii), we also see that $\T$ contains at most $|\E| = e$ edges.
\qed 
\end{proof}

Combining the above arguments, we have the main theorem of this section. 

\begin{proposition}\label{thm:fwd:bwd:dfs}
Assume the standard path ordering $\Pi^\pos_\lex = (\leqpos, \leqlex)$.
We can perform the following tasks in $O(e)$ worst-case time and words of space on the ordered CDAWG $G$ under $\Pi^\pos_\lex$ for $T$: 
(1) Enumerating all $(-)$-certificates in $\ESX[-]$ in the lexicographic order $\leqlex$ by the ordered ordered DFS of $\T[-]$. 
(2) Enumerating all $(+)$-certificates in $\ESX[+]$ in the text order $\leqpos$ by the backward ordered DFS of $\T[+]$.
\end{proposition}

\begin{proof}
By \cref{lem:sp:spanning:trees}, we can enumerate all $(-)$-certificates in the lexicographic order $\leqlex$ by the standard ordered DFS of $\T[-]$ starting from the root of $G$ and going downwards it iteratively following $\EP[-]$ edges; When it encountered an $\ES[-]$ edge, it report it and backtracks. At any node $v$, its outgoing edges are visited in the order of $\leqe[+] = \leqelex$.
We remark that the DFS traverses the same edge at most once because in the case that more than one edge of $\T$ meet at the same node on $G$, exactly one of them is $(-)$-primary. Therefore, the DFS follows only exactly one $\EP$-edge, and backtracks with all the remaining $\ES$-edges.
By symmetry, we can enumerate all $(+)$-certificates in the text order $\leqpos$ by the backward ordered DFS of $\T[+]$ starting from the sink of $G$, going upwards following $\EP[+]$ edges by selecting incoming edges in the order of $(\leqe[-]) = (\leqepos)$. In either case, the DFS over $\T$ traverses at most $O(|\T|) = O(e)$ edges.
\qed 
\end{proof}
  
Concerning to \cref{thm:fwd:bwd:dfs}, from (2) of \cref{lem:canonical:factor},
we can put the pointer from each discovered $(-)$-certificate $f_-$ to its $(+)$-counterpart $f_+$ such that $\cano[-](f_-) = \cano[+](f_+)$ in amortized $O(1)$ time per certificate and \textit{vice versa}.



\section{Computing Run-Length BWT}
\label{sec:algbwt}

\subsubsection{Characterizations}
\label{sec:algbwt:characterization}
Given the BWT for a text $T$, the set of all \textit{irreducible ranks} is given by the set 
\begin{inalign}
  I_{BWT} := \sete{ i \in [n] \mid BWT[i] \not= BWT[i-1] }
  \subseteq [n].  
\end{inalign}
Then, we define the set $QI_{BWT}$ of all \textit{quasi-irreducible ranks} by the set
\begin{inalign}
  QI_{BWT} &:= \sete{ \Rnk(S) \mid S \in \CS }
  \subseteq [n],   
\end{inalign}
Obviously, $|QI_{BWT}| \le e$ since $\Rnk$ is a bijection and $|\CS[]| \le e$.
By assumption, we can show that $I_{BWT} \subseteq QI_{BWT}$. Consequently, $QI_{BWT}$ satisfies the following \textit{interpolation property}. 

\begin{toappendix}
\begin{lemmarep}\label{lem:bwt:ibwt:qibwt:inclusion}
Under $(\lequp, \leqdw) = (\leqpos, \leqlex)$, 
$I_{BWT} \subseteq QI_{BWT}$. 
\end{lemmarep}

\begin{proof}
  First, we see that $QI_{BWT} = \dom({\RSA(CS_-)})$ by definition and $\dom({\RSA(CS_-)}) = \dom({\RSA(CS_+)})$ by~\cref{lem:cano:euivalence:up:down}. Thus, we have $QI_{BWT} = \dom({\RSA(CS_+)})$, and
  we will show that $I_{BWT}\subseteq \dom({\RSA(CS_+)})$, or equivalently, for any $i \in I_{BWT}$, there exists some $(+)$-canonical suffix $S$ such that $\Rnk(S) = i$. 
  We remark that from \cref{lem:qbwt:firstrank}, we will show that if $i \in I_{BWT}$ then $i \in QI_{BWT}$ holds for any rank $i \in 2..n$.
Let us start with any element $i \in I_{BWT}$ such that $i > 1$ and $BWT[i] \not= BWT[i-1]$. By assumption, there exists a pair of lexicographically adjacent suffixes, $T_q = T_{SA[i-1]}$ and $T_p = T_{SA[i]}$ such that $T_q \lesslex T_p$. Since this pair satisfies the pre-condition of \cref{lem:tech:adjacent:suffixes:leqdw} under $\lessdw = \lesslex$, we can put  $T_r = P X_r D_r$ for $\forall r = p,q$ such that 
\begin{enumerate*}[(i)]
\item $T_q$ and $T_p$ are root-to-sink paths with the unique branching node $v$;  
\item $P = T_p[1..\ell] = lcp(T_{SA[i]}, T_{SA[i-1]})$ is the common suffix of $T_q$ and $T_p$ starting from $\root(G)$ with its locus $v$; 
\item $X_r$ is the label of branching edge $f_r = (v, X_r, w_r) \in \ESD$ at node $v$ such that $f_q \lessein f_p$; Moreover, we see $f_p \in \ESU$ (c1); 
\item $D_p = \min_{\leqlex} \U[+](w_p) = \reprdw(w_p)$ (c2). 
\end{enumerate*}
Now, we will show that $P = \reprup(v)$ holds. Since $P$ is a common prefix of $\set{T_q, T_p} = \set{T_{SA[i-1]}, T_{SA[i]}}$,
the ranks of the suffixes $\set{T_{SA[i-1]}, T_{SA[i]}}$ belong to the same SA-interval for $P$, i.e., $\set{i, i-1}\subseteq I(P)$. 
    Then, we observe that $P$ is left-maximal; Otherwise, the BWT-interval for $I(P)$ consists of the same symbol, say $c \in \Sigma$, however, this contradicts the assumption that $BWT[i] \not= BWT[i-1]$.
    As seen in \cref{sec:tech:edge:class}, $P$ is left-maximal if and only if $P = \reprup(v)$. Thus we have $P = \reprup(v)$~(c3).
    Combining (c1)--(c3), we see that
    $T_p$ is $(+)$-canonical by certificate edge $f_p = (v, X_p, w_P) \in \ESD$, belonging to $CS_+$. 
    By the above arguments, the claim follows. 
\qed
\end{proof}
\end{toappendix}

\begin{lemmarep}[interpolation property]
  \label{lem:bwt:prop:interpolation}
Under $(\lequp, \leqdw) = (\leqpos, \leqlex)$, 
if $i_* \not\in QI_{BWT}$, $BWT[i] = BWT[i-1] \in \Sigma$ holds for $\forall i \in [n]$. 
\end{lemmarep}

\begin{proof}
By assumption, we can show that $I_{BWT} \subseteq QI_{BWT}$. 
Therefore, if any rank $i$ is quasi-reducible w.r.t.~$QI_{BWT}$, then
$BWT[i] = BWT[i-1]$.
    \qed
\end{proof}


\setlength{\textfloatsep}{.5\baselineskip}
\begin{algorithm}[t]
\renewcommand{\-}{\mbox{\textrm{-}}}
\caption{The algorithm for computing the quasi-irreducible BWT for $T[1..n]$ from the CDAWG $G$ for $T$ stored in read-only memory. 
}
\label{algo:recbwt}
\kw{Procedure} \algo{RecRBWT}$(v)$\;
\Begin{
    \iIf{$\Out(v) = \emptyset$}{
      \Return $(\mbox{`$\daller$'}, 1)$
      \Comment*{Case: trivial suffix. $T[n] = \mbox{`$\daller$'}$}
    }
    \Else (\Commentblock{Case: non-trivial suffix}) {
      \For{\textbf{each} $f = (v, X, w) \in \Out(v)$ in order $\leqelex$ compatible to $\leqlex$}{
        \uIf (\Commentblock{Case: $(-)$-primary}) {$\isprimary_{-}(f)$} {
          $RBWT' \gets \algo{RecRBWT}(w)$\;
        }
        \Else (\Commentblock{Case: $(-)$-secondary}) {
          $c \gets \op{precsym}(f)$; 
          $\ell\gets \op{nleaves}(\dst(f))$;
          $RBWT' \gets (c, \ell)$\; 
        }
        $RBWT \gets RBWT \circ RBWT'$\;
      }
      \Return $RBWT$\; 
    }
}
\end{algorithm}

\subsubsection{Algorithm.}
In \aref{algo:recbwt}, we present the recursive procedure that computes the quasi-irreducible BWT for text $T$ from an input CDAWG $G$ for $T$ stored in read-only memory or the self-index $G = \CDAWGd$ when it is invoked with $v = \root(G)$ and $RBWT = \eps$.
Let $F = (P_1, \dots, P_h)$, $h \le e$, be the sequence of all $(-)$-search path of $\SP[-]$ sorted in the lexicographic order $\leqlex$ of string labels with the index $i_*$ of the trivial $(-)$-path. 
Let $\sig I = (I_1, \dots, I_h)$ be the associated sequence of SA-intervals such that $I_i = [sp(P_i)..ep(P_i)] \subseteq [n]$ for all $i \in [h]$.
From \cref{lem:sp:freeness}, we can show that $\sig I$ forms an ordered partition of $[n]$, namely, the elements of $\sig I$ are ordered in $\leqlex$ and any suffix falls in exactly one interval of $\sig I$. 

\begin{toappendix}
\begin{lemmarep}\label{lem:bwt:ordered:part}
  $\sig I$ forms an ordered partition of $[n]$. 
\end{lemmarep}

\begin{proof}
  We show that
  \begin{enumerate*}[(i)]
  \item $\bigcup_{i} I_i = [n]$, and 
  \item $I_i \cap I_j = \emptyset$ for any $i < j$,
    and moreover, 
  \item $sp(P_1) = 1$, $ep(P_h) = n$, and $sp(P_i) = ep(P_{i-1}) + 1$ for all $i \in [2..n]$.
\end{enumerate*}
  Let $S \in \Suf$ be any suffix with $\Rnk(S) = k \in [n]$. Then, we can show that there exists some $X \in \SP[-]$ such that $X$ is a prefix of $S$. If $X$ is $i$-th smallest member of $\SP[-]$ under $\leqlex$, we have $\Rnk(S) \in I(X) = I_i$. This implies claim (i). The claim (ii) immediately follows from \cref{lem:sp:freeness}. Since $\sig I$ is an ordered partition of $[n]$ from (i) and (ii), claim (iii) follows.
\qed   
\end{proof}
\end{toappendix}

Now, we give a characterization of the BWT in terms of the $(-)$-search paths for the canonical suffixes. 

\begin{lemmarep}
  \label{lem:bwt:characterize}
  Let $T[1..n] \in \Sigma^n$ and $BWT[1..n] \in \Sigma^n$ be the BWT for $T$. 
  \begin{enumerate}
  \item $BWT[1..n]  = BWT[I_1]\circ \dots \circ BWT[I_h]$. 
  \item For each $i \in [h]$,
    the conditions (i) and (ii) below hold: 
  \begin{enumerate*}
  \item If $i = i_*$, $P_i$ is a trivial $(-)$-search path, and then $I_{i_*}$ is a singleton and $BWT[I_{i_*}] = T[n] = \mathdaller$.
  \item
    If $i\not= i_*$, 
    $P_i$ is a non-trivial $(-)$-search path with certificate $f \in \ES[-]$. Then, $BWT[I_1] = c^\ell$, 
    where
    $c := T[p-1]$, 
    $\ell := |I_i|$, and 
    $p = \Pos(\cano(f))$.  
  \end{enumerate*}
  \end{enumerate}
\end{lemmarep}

\begin{proofsketch}
  Claim (1) immediately follows from the definition of $\sig I$. Claim (2) is obvious since trivial $(-)$-search path is the longest suffix $T[1..] = T$ itself.
  (3) Suppose that $X$ is a non-trivial $(-)$-search path with locus $v$. Then, $X$ is not equal to $\repr[-](v)$. 
  As seen in \cref{sec:tech:edge:class}, it follows that $X$ is not left-maximal in $T$. Therefore, there exists some $c \in \Sigma$ that precedes all start positions of $X$ in $T$. $I_i$ gives the number of leaves below $v$. Therefore, Claim (3) is proved. \qed 
\end{proofsketch}

\subsubsection{Case with a read-only text.} 
Suppose that the read-only text $T[1..n]$ is available.
By \cref{lem:bwt:characterize}, we see that all BWT-intervals $BWT[I_i]$ but $BWT[I_{i_*}]$ are equal-symbol runs with length $|I_i|$, while $BWT[I_{i_*}]$ is the singleton $\mathdaller$.
Since the position $p$ of the lex-first suffix $\cano(f)$ in $I_i$ can be obtained in constant time by $p = \Pos(\cano(f))$, the preceding symbol, denoted $\precsym(f) := T[p-1]$, can be obtained in constant time in the case of a read-only text.
Concatenation of two run-length encodings can be done in $O(1)$ time by maintaining the symbols at their both ends.  
Hence, we can construct the RLBWT of size $r \le e$ in $O(e)$ worst-case time and $O(e)$ words of space using \aref{algo:recbwt}. 

\subsubsection{Extension to the case without a text.} 
Next, we consider the case that input is the self-index version of $\CDAWGo$ without access to the text $T$. 

\begin{figure}[t]
  \centering  
  \begin{subfigure}[t]{0.49\textwidth}
    \centering
    \nofbox{\includegraphics[width=0.95\textwidth]{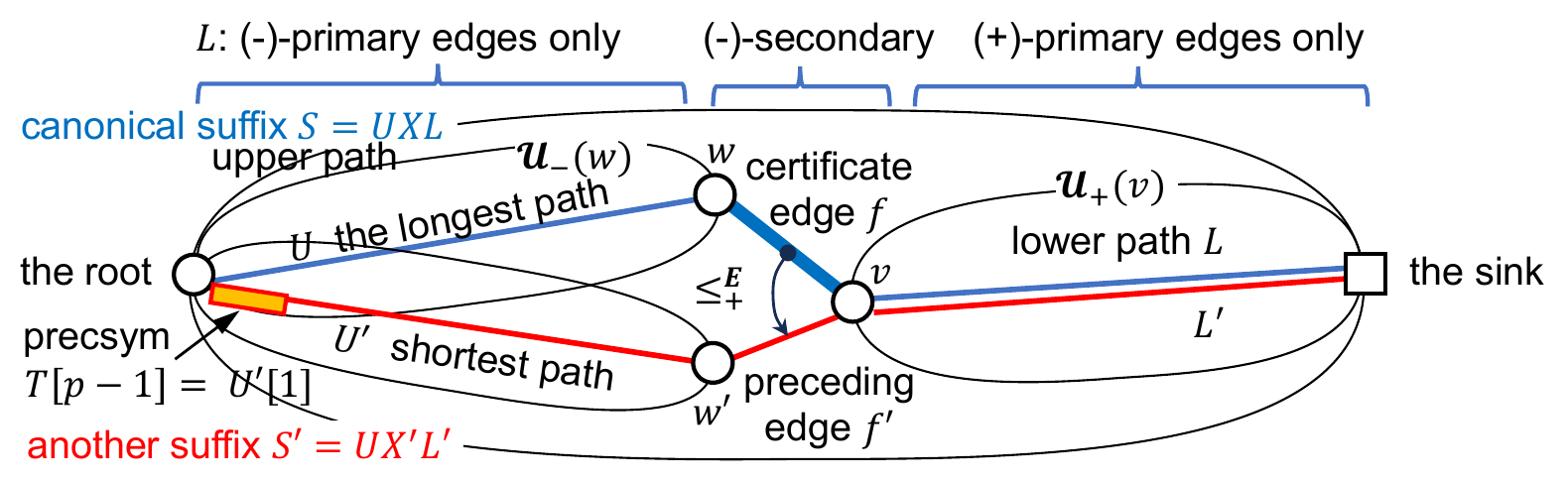}}
    \caption{BWT: Computing a preceding symbol}\label{sub:one}
  \end{subfigure}
  \begin{subfigure}[t]{0.49\textwidth}
    \centering
    \nofbox{\includegraphics[width=0.95\textwidth]{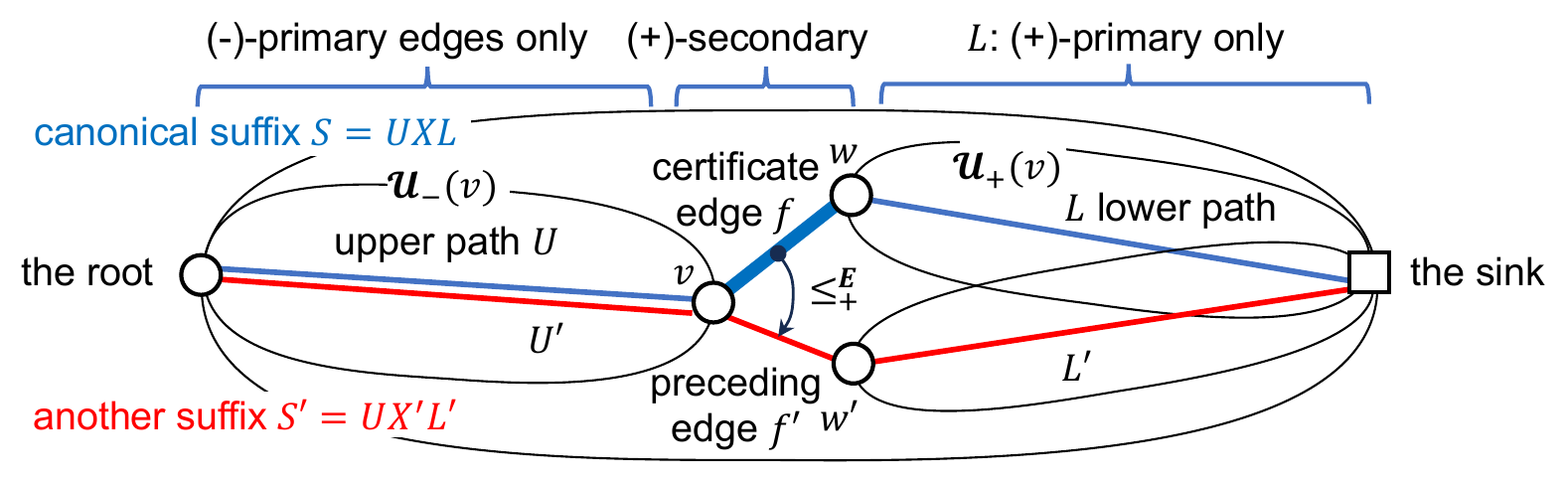}}
    \caption{GLPF: Computing a previous factor}\label{sub:two}
  \end{subfigure}
  \caption{Computation of run-length BWT and quasi-irreducible GLPF}
\end{figure}

\begin{lemmarep}[computing a preceding symbol]
  \label{lem:bwt:weiner:symbol}
  Given $G = \CDAWGd$, the set of $O(e)$ preceding symbols $\op{precsym}(f) := T[p - 1]$ of the non-trivial canonical suffix $S = \cano(f)$ for all certificate edges $f \in \ESU$ can be computed 
  in $O(e)$ worst-case time and words of space. 
\end{lemmarep}

\begin{proofsketch}
  Let $S$ has position $p$ and . 
  Since $S$ is non-trivial, there exists the predecessor $S'$ of $S$ such that $|S'| = |S| + 1$. Furthermore, there exists another incoming edge $f'$ such that $f' \leqepos f$, $f = (v, X, w)$, and $f' = (v', X', w)$ hold. Then, $S'$ and $S$ meet at node $v$ going through $f'$ and $f$, and thus, they can be factorized as $S' = U' X' L$ and $S = U X L$ (See \cref{sub:one}). Since $U'$ is the shortest in $\U[-](w')$, it follows from \cref{lem:preprocessing} that $\op{precsym}(f) = U'[1]\in \Sigma$ can be computed by $\op{fstsym\-shortest}(w')$ using $O(e)$ preprocessing and space. \qed   
\end{proofsketch}

\begin{proof}
  In \cref{fig:weiner:symbol}, we show the idea of the proof. 
  Let $f = (v, X, w)$ be any $\ESU$-edge. Then,
  there is a pair $(S', S)$ of suffixes in $\Suf$ such that $|S'| = |S| + 1$, where $S'$ is the predecessor of $S$ in $\Suf$. Suppose that $f'$ and $f$ are the pair of incoming edges of $w$ through which $S'$ and $S$ come from the root to $w$, where $f' = (v', X', w)$. That is, we can put $S' = U'X'D$ and $S = UXD$, where $U' = \shortestup(v')$, $U = \reprup(v)$, and $D = \reprdw(w)$.
  If $f$ is $(-)$-secondary, such $f'$ always exists.
  If the start position of $S$ is $p$ in $T$, the start position of $S'$ is $p-1$. Therefore, we can obtain $\op{precsym}(f) = T[p-1]$ as $T[p-1] = S'[1] = U'[1]$. This can be easily computed by the reverse depth-first search from the sink. The details are omitted. 
\qed
\end{proof}

\begin{theoremrep}\label{thm:algbwt:complexity}
  Let $T[1..n]$ be any text over an integer alphabet $\Sigma$. Given a self-index version of $\CDAWGo$ without a text, \aref{algo:recbwt} constructs the RLBWT of size $r \le e$ in $O(e)$ worst-case time and $O(e)$ words of space. 
\end{theoremrep}

\section{Computing Irreducible $\GLPF$ Arrays}
\label{sec:alglpf}

\subsubsection{Characterizations.}
\label{sec:alglpf:characterization}
In order to treat the PLCP and LPF arrays uniformly,
we introduce their generalization, called the quasi-irreducible $\GLPF$ array for a text $T$ parameterized by $\leqp[+] \in \set{\leqlex, \leqpos}$ according to~\cite{navarro:ochoa:2020approx,crochemore2008computing:lpf}.

\begin{definition}
  \label{def:glpf:array}
  The \textit{generalized longest previous factor (GLPF) array} for a text $T[1..n]$ under $\leqdw$ is the array $\GLPFd[1..n] \in \nat^n$ such that for any $p \in [n]$,
\begin{inalign}
  \GLPFd[p] &:= \max (\;\set{\:lcp(T_{p}, T_{q}) \mid T_q \myprec T_p, \:q \in [n]\:}\cup\set{0}\:).
  \label{eqn:glpf:array}
\end{inalign}
\end{definition}

\begin{lemmarep}\label{lem:glpf:plcp:lpf}
For any text $T$, we have $PLCP$ $=$ $\GLPF_{\leqlex}$ and $LPF = \GLPF_{\leqpos}$. 
\end{lemmarep}

\begin{proof}
  (1) The proof is similar to~\cite{navarro:ochoa:2020approx,crochemore2008computing:lpf}.
  First, consider the case for $PLCP$. Let $PLCP[p] = \ell$.
  By definition of $PLCP$,
  $\ell := lcp(T_{i}, T_{i-1})$ such that $p = SA[i]$ and $i = ISA[p]$.
  Then, we can show that for any $p\in [n]$ and $i = ISA[p]$,
  $q = SA[i-1]$ maximizes $lcp(T_p, T_q)$ over all $q \in [n]$ such that $T_q \leqlex T_p$. Hence, $PLCP[p] = \GLPF_{\leqlex}[p]$ for any $p \in [n]$.
  On the other hand, the case for $LPF$ is obvious since the definitions of  $LPF$ and $\GLPF_{\leqpos}$ are identical when we put $\leqdw = \leqpos$.
Hence, the lemma is proved. 
\qed 
\end{proof}

Now, we introduce the quasi-irreducible GLPF array as a subrelation $\qir{\GLPF}$ of $\GLPF$ indexed in $\leqpos$. 
Under $\Pi^\pos_\pos = (\leqpos, \leqp[+])$, we define the quasi-irreducible GLPF array by the binary relation 
\begin{inalign}
  \qir{\GLPF} &:= \sete{ (\Pos(S), \val(S)) \mid S \in \CS[] } 
  \subseteq [n]\by\nat.  
\end{inalign}
where $\val(S) = \GLPFd[\Pos(S)]$ and 
$QI_{\GLPF} := \sete{ \Pos(S) \mid S \in \CS[] }$
is the set of \textit{quasi-irreducible} ranks.
Since $|\CS[]| \le e$, $\qir{\GLPF}$ has size $e$.
We observe that $GLPF[p] = 0$ implies $p \in QI_{GLPF}$. 
Then, $\GLPFd$ satisfies the interpolation property below. 

\begin{toappendix}
\begin{remarkrep}\label{lem:glpf:firstrank}
  For any $T$, $1 \in QI_{\GLPFd}$ holds. 
\end{remarkrep}

\begin{proof}
  By symmetry, we can show the claim similarly to one for \cref{lem:qbwt:firstrank}. 
  Suppose that $S_1 = \min_{\leqpos} \Suf$ is the longest suffix in $\Suf$. Consequently, we have $\Pos(S_1) = 1$.
  There are two cases below.
  If $S_1$ consists only of $\EP[+]$ edges, $S_1$ is the unique trivial canonical suffix $\reprup(\sink(G))$, and we are done.
  Otherwise, $S_1$ contains some $\ES[+]$ edges.
  Then, we can put
  $S_1 = U_* X D_*$ with
  $U_* \in (\E(G))^*$ and 
  $D_* \in (\EP[+])^*$
  for the bottom $\ES[+]$ edge $f_+ = (w, X, v)$. 
  Since $S_1$ is the longest, all edges in $U_*$ must be $\EPU$-edges.
  Thus, we have $U_* = \reprup(v)$.
  On the other hand, we have $D_* = \reprdw(w)$ since $D_* \in (\EP[+])^*$.
  Combining above arguments, we see that $S_1 = \cano(f_+)$ belongs to $\CS[+]$. Hence, $\Pos(S_1) = 1 \in QI_{\GLPFd}$. This completes the proof.
\qed 
\end{proof}
\end{toappendix}

\begin{propositionrep}[interpolation property]
  \label{lem:glpf:prop:interpolation}
For any position $p \in [n]$, if $p \not\in QI_{\GLPFd}$ then $\GLPFd[p] = \GLPFd[p-1] - 1$ holds.
Consequently, $PLCP$ and $LPF$ satisfy the same interpolation property w.r.t.~$QI_{GLPF}$. 
\end{propositionrep}

  
  
  

\begin{proof}
  First, we show the next claim.

\par \textit{Claim~1}: $\GLPFd[p] \ge \GLPFd[p-1] - 1$ for any position $p \in [2..n]$. 

\par (Proof for Claim~1)
We show that $\GLPFd[p+1] \ge \GLPFd[p] - 1$ for any $p \ge 1$.
Let $\ell := \GLPFd[p]$. 
If $\ell = 0$, we are done because $\GLPFd[p+1] \ge 0$ always holds.
Otherwise, $\ell \ge  1$ and $T[p] = T[q]$ (*) holds.
  From the extensibility of $\leqdw$,
  $T_q \leqdw T_p$ implies that $T_{q+1} \leqdw T_{p+1}$.
  Since $\GLPFd[p+1]$ is the maximum of $lcp(T_{p+1}, T_{q'})$ for all $q'\in [n]$ such that $T_{q'} \leqdw T_{p+1}$, it
  follows from (*) that $\GLPFd[p+1] \ge lcp(T_{p+1}, T_{q+1}) \ge lcp(T_p, T_q) - 1 = \GLPFd[p] - 1$. 
(End of the proof for Claim~1)
  
  The remaining task is to show the next claim. 
  \par \textit{Claim~2}: $\GLPFd[p-1] \ge \GLPFd[p] + 1$ for all $p > 1$.

  \par (Proof for Claim~2)
  By definition, $\GLPFd[p] = lcp(T_p, T_q) = \ell$ and $T_q \leqdw T_p$ for some $q \in [n]$. Thus, $T_p, T_q$ have the common prefix $P := T[p..p+\ell-1] = T[q..q+\ell-1]$ of length at least $\ell$. 
  By \cref{lem:glpf:prop:leftmaximality}, if $p \not\in QI_{\GLPFd}$ then $P$ is not left-maximal.
If $P$ is not left-maximal, then there exists some symbol $c \in \Sigma$ such that $BWT[I(P)] = c^{|I(P)|}$. Since both $T_p$ and $T_q$ are prefixed by $P$, they fall into $I(P)$. Therefore, we now have $T[p-1] = T[q-1] = c$. This implies that
$\GLPFd[p-1] = lcp(T_{p-1}, T_{q-1}) = 1 + lcp(T_{p}, T_{q}) = 1 + \ell = LCP[p] + 1$.
(End of the proof for Claim~2)
From Claim~1 and Claim~2, we showed the lemma.
\qed 
\end{proof}

\begin{lemmarep}[characterization of $\GLPFd$ value]
  \label{lem:glpf:prop:leftmaximality}
  For any pair $(p, \ell) \in [n]\by\nat$, the conditions (a)--(c) below are equivalent (See \cref{sub:two}): 
\begin{enumerate}[(a)]  
\item $(p, \ell) \in \qir{\GLPF}$.  
\item $\GLPFd[p] = \ell$ and $T_p[1..\ell] = T[p..p+\ell-1]$ is left-maximal in~$T$.  
\item For some $S \in \CS[]$, $p = \Pos(S)$.
  Also, 
  if $S$ is $(+)$-trivial $\ell = 0$, and 
  if $S$ is $(+)$-nontrivial, it has the form
  $S = \cano[+](f) = \reprup(w)\cdot X\cdot \reprdw(v)$
  for some $(+)$-certificate $f = (w, X, v) \in \ESD$ and
  $\ell = |\reprup(w)|$ holds.  
\end{enumerate}
\end{lemmarep}

\begin{proof}
  Let $p_* \in [n]$ be the special position of the trivial canonical suffix $\vobj[+]S = \reprdw(\root(G))$.
  Let $S = T_p$ with $p = \Pos(S)$. 
  Suppose that $\GLPFd[p] = \ell$ with $p \not= p_*$. Then, there exists some other suffix $T_q$ with $T_q \leqdw T_p$ that gives $lcp(T_p, T_q) = \ell$. Let $w$ be the branching node of $T_p$ and $T_q$ and $f_p, f_q \in \Out(v)$ be outgoing edges through which $T_p$ and $T_q$ go, respectively. Since outgoing edges in $\Out(w)$ are ordered by the edge order $\leqeout$ compatible to $\leqdw$, we have $f_q \lesseout f_p$ (*); Furthermore, since $w$ is the deepest branching point over all $T_q \in \Suf$, $f_p$ is the lowest $\ESD$-edge in $T_p$ (**).
Now, we consider the equivalence between (a)--(c). 
  
(c) $\Rightarrow$ (a): The claim follows immediately in both cases, trivial and non-trivial, by definition of $QI_{\GLPFd}$

(a) $\Rightarrow$ (b): Suppose that  $p \in QI_{\GLPFd}$. Then, there exists some $(+)$-canonical suffix $S = \cano(f)$ for $f \in \CE[+]$. 
(a.i) Trivial case: $f = \vobj[+]f$ and $S = T_{p} = \reprdw(\root(G))$. Then, there is no suffix $T_q$ such that $T_q \lessdw T_p$, and thus, we see that $\ell = \GLPFd[p] = 0$ by definition. Since the empty string $\eps$ is always left-maximal in $T$. Therefore, claim (b) holds. 
(a.ii) Non-trivial case: Suppose that $p\not= p_*$ and $S = T_p = \reprup(w)\cdot X\cdot \reprdw(v)$ for some edge $f_* = (w, X, v) \in \ESD$ such that $\id{pos}(T_p) = p$.
By similar discussion above (*) shows that $f_p = f_*$ is the lowest $\ESD$-edge in $T_p$ and $f_q \lessdw f_p = f_*$ holds for some suffix $f_q$. It immediately follows that $\ell = lcp(T_p, T_q) = |\reprup(w)|$ and $\reprup(w) = T_p[1..\ell]$. As seen in \cref{sec:tech:edge:class}, $\reprup(w)$ is left-maximal.

(b) $\Rightarrow$ (c):
(b.i) Trivial case: Obviously, $\ell = 0$ holds. 
(b.ii) Non-trivial case: $p\not=p_*$ and $S$ contains some $\ESD$-edges. 
Then, by claim (**), we can partition $T_p$ as 
$T_p = U X D$ for the deepest $\ESD$-edge $f_p = (w, X, v)$ in $T_p$. Since the factor $D$ contains no $\ESD$-edges by assumption, we have $D = \reprdw(v)$.
Since $f_p\in \ESD$, there must exist some edge $f_q \in \Out(w)$ with $f_q \lessein f_p$, and thus, some suffix $T_q$ such that $T_q \lessdw T_p$.
If $\ell = lcp(T_p, T_q)$, $T_p[1..\ell]$ must be the prefix $U$ since $w$ is the branching point of $T_p$ and $T_q$. 
Furthermore, if $w$ is the lowest left-branching position, $\ell$ is the largest among such suffixes $T_q$. Thus, we have $\GLPFd[p] = \ell$. If $T_p[1..\ell] = U$ is left-maximal, we have $U = \reprup(w)$
as seen in \cref{sec:tech:edge:class}. By discussions above, we have $T_p = \reprup(w)\cdot X\cdot \reprdw(v) = \cano(f_p)$. Since $|\reprup(w)| = |U| = \ell$, the claim is proved. 
Combining the above arguments, we showed the lemma.
\qed 
\end{proof}

\setlength{\textfloatsep}{.5\baselineskip}
\begin{algorithm}[t]
\renewcommand{\-}{\mbox{\textrm{-}}}
\SetArgSty{textrm}
\SetCommentSty{textit}
\SetKwInput{KwGlobal}{Global variable}
\caption{
  The algorithm for computing the quasi-irreducible GLPF array for a text $T$ from the CDAWG $G$ for $T$ or its self-index. 
}
\label{algo:recirrglpf}
\kw{Procedure} \algo{QIrrGLPF}$(v, QGL)$
\Comment*{Assume path orderings $\Pi = (\leqpos, \leqdw)$}
\Begin{
    \uIf (\Commentblock{Case: trivial suffix at the root}) {$\In(v) = \emptyset$}{
      $QGL \gets QGL \circ (1, 0)$ 
    }
    \Else (\Commentblock{Case: non-trivial suffix at branching node}) {
    \For{each $f = (w, X, v)$ in order $\leqepos[-]$ compatible to $\leqpos$
    } {
      \uIf (\Commentblock{Case: $(+)$-primary}) {$\isprimary_{+}(f)$} { 
        $\algo{QIrrGLPF}(w, QGL)$
      }
      \Else (\Commentblock{Case: $(+)$-secondary}) {
        $\ell \gets |\reprup(w)|$; 
        $p \gets n+1- |\reprup(w)| - |X| - |\reprdw(v)|$\; 
        $QGL \gets QGL\circ (p, \ell)$\Comment*{$output: \GLPF[p] = \ell$}
      }
    }
    }
  }
\end{algorithm}

\subsubsection{Algorithm.}
\label{sec:alglpf:alg}
In \aref{algo:recirrglpf}, we present the recursive procedure for computing the quasi-irreducible GLPF array, $\qir{\GLPF}$, of size $e$ for a text $T[1..n]$ from the self-index $G = \CDAWGd$ of size $O(e)$ under a parameter pair $\Pi = (\lequp, \leqdw)$, when it is invoked with $v = sink(G)$ and $QGL = \eps$.

\begin{theoremrep}\label{thm:alglpf:complexity}
  Let $T[1..n]$ be any text over an integer alphabet $\Sigma$. Given a self-index version of $\CDAWGo$ without a text, \aref{algo:recirrglpf} constructs the quasi-irreducible $\GLPFd$ array for $T$ of size $e$ in $O(e)$ worst-case time and $O(e)$ words of space. 
\end{theoremrep}


By a simple procedure as in~\cite{crochemore2008computing:lpf,navarro:ochoa:2020approx}, we can easily compute either the lex-parse from $\GLPF_{\leqlex}$ or the LZ-parse from $\GLPF_{\leqpos}$ in linear time in combined input and output sizes. Hence, the next theorem follows from \cref{lem:glpf:plcp:lpf} and \cref{thm:algbwt:complexity}. 

\begin{theoremrep}\label{thm:algglpf:parse:lex:lz}
Let $T[1..n]$ be any text over an integer alphabet $\Sigma$.
The lex-parse of size $2r = O(e)$ and the LZ-parse of size $z \le r$ of $T$ can be computed from a self-index version of $\CDAWGo$ without a text for the same text in $O(e)$ worst-case time and words of space. 
\end{theoremrep}

\begin{proof}
We first convert $\CDAWGo$ or $\CDAWGd$ into $\qir\GLPF_{\leqlex}$ and $\qir\GLPF_{\leqpos}$ by \cref{thm:algglpf:main}, and then, we convert $\qir\GLPF_{\leqlex}$ and $\qir\GLPF_{\leqpos}$ into the lex-parse and LZ-parse of $T$, respectively, by~\cref{lem:glpf:plcp:lpf} and \cref{lem:glpf:greedy:parse}. 
Since the sizes of the lex-parse and LZ-parse are bounded from above by $v = 2r \le e$ and $z \le e$, respectively, all computations can be done in $O(e)$ time and space by \cref{thm:algglpf:main} and \cref{lem:glpf:greedy:parse}. This completes the proof. 
\qed  
\end{proof}


\subsection*{Acknowledgments}
The authors thank the anonymous reviewers for their helpful comments which greatly improved the quality and presentation of the paper.
The first author is also grateful to Hideo Bannai for providing information on the literature on sublinear time and space conversion between text indexes, and to Mitsuru Funakoshi for discussing the sensitivity of text indexes for morphic words. 




\end{document}